\documentclass[a4paper, 10pt]{article}

\usepackage{amsmath, amssymb, amsthm}
\usepackage[CJKbookmarks=true,colorlinks,linkcolor=black,anchorcolor=black,citecolor=black]{hyperref}
\usepackage{CJK}
\usepackage{multirow}
\usepackage{booktabs}
\usepackage{appendix}
\usepackage{tikz}
\usepackage{graphicx, xcolor}
\usepackage{overpic,subfig, float}
\usepackage{pgfplots}
\usetikzlibrary{external}
\usepackage[stable]{footmisc}

\newcommand\bbR{\mathbb{R}}
\newcommand\bbS{\mathbb{S}}
\newcommand\bbN{\mathbb{N}}

\def\+#1{\boldsymbol{#1}}
\newcommand\ang[1]{\left\langle {#1} \right\rangle}

\newcommand\mH{{\mathcal H}}
\newcommand\md{{\mathrm{d}}}

\newcommand\tn[1]{\mathcal{N}{({#1})}}
\newcommand\tni[1]{\mathcal{N}_1{({#1})}}
\newcommand\tnj[1]{\mathcal{N}_2{({#1})}}

\newcommand\od[2]{\dfrac{\mathrm{d} {#1}}{\mathrm{d} {#2}}}

\newcommand\ifup[1]{}

\newcommand\onlinecite[1]{{\cite{{#1}}}}

\begin{document}
\begin{CJK*}{UTF8}{gkai}
\theoremstyle{plain}
\newtheorem{lemma}{Lemma}
\newtheorem{theorem}{Theorem}
\newtheorem{cor}{Cor}
\newtheorem{prop}{Prop}
\newtheorem{remark}{Remark}

\theoremstyle{definition}
\newtheorem{defn}{Definition}
\newtheorem{exmp}{Example}

	\newenvironment{@abssec}[1]{%
       \vspace{.05in}\footnotesize
	   \noindent
         {\upshape\bfseries #1 }\ignorespaces
	 }

\newcommand\keywordsname{Keywords}
\newcommand\AMSname{AMS subject classifications}
\newcommand\Ackname{Acknowledgments}
\newcommand\datav{Data Availibility}
\newenvironment{AMS}{\begin{@abssec}{\AMSname}}{\end{@abssec}}
\newenvironment{Ack}{\begin{@abssec}{\Ackname}}{\end{@abssec}}
\newenvironment{datava}{\begin{@abssec}{\datav}}{\end{@abssec}}

\bibliographystyle{plain}

\title{Slip and Jump Coefficients for General Gas-Surface 
	Interactions According to the Moment Method}
\author{Ruo Li\thanks{CAPT, LMAM \& School of Mathematical Sciences,
    Peking University, Beijing 100871, China, email: {\tt
      rli@math.pku.edu.cn}.} \and Yichen Yang\thanks{School of
    Mathematical Sciences, Peking University, Beijing 100871, China, email:
    {\tt yichenyang@pku.edu.cn}.}
}

\maketitle{}

\begin{abstract}
	We develop a moment method based on the Hermite series of 
	arbitrary order to calculate viscous-slip, thermal-slip,
	and temperature-jump coefficients for general gas-surface
	scattering kernels. Under some usual assumptions of scattering
	kernels, the solvability is obtained by showing the positive
	definiteness of the symmetric coefficient matrix in the
	boundary conditions. For gas flows with the Cercignani-Lampis
	gas-surface interaction and inverse-power-law intermolecular
	potentials, the model can capture the slip and jump coefficients
	accurately with elegant analytic expressions. On the one hand, 
	the proposed method can apply to the cases of arbitrary order 
	moments with increasing accuracy. On the other hand, the explicit
	formulae for low-order situations are simpler and more accurate
	than some existing results in references. Therefore, one may 
	apply these formulae in slip and jump conditions to improve 
	the accuracy of macroscopic fluid dynamic models for gas flows.
\end{abstract}

\section{Introduction}
The rarefied gas flow effects such as the velocity slip, temperature
jump, and thermally induced flows near the solid wall, are fundamental
issues in micro-electro-mechanical systems (MEMS) and low-density
hypersonic aerodynamics. \cite{MEMS2005,Cao2009,Zhang2012,Akh2023}
The rarefaction is often measured by the Knudsen
number, i.e., the ratio of the mean free path to the 
characteristic length. As the Knudsen number goes
larger, the continuum assumption breaks down, and the traditional
Navier-Stokes equations are no longer
applicable. One may impose slip and jump boundary conditions to
enlarge the application scope of the Navier-Stokes equations,
where the slip and jump coefficients play a crucial role
in the accuracy of the macroscopic fluid dynamic equations.
\cite{Tropea2007,Akh2023}

From a statistical standpoint, the Boltzmann 
equation can describe the rarefied gas flows, with a general
scattering kernel to specify the gas-surface interaction. \cite{CC1989}
The most popular scattering kernel may be the Maxwell 
diffuse-specular kernel, which has one accommodation coefficient
(AC) to parameterize the roughness of the solid wall. 
However, as shown in experiments and molecular dynamics (MD)
simulations, \cite{Yama1999,Tropea2007,Sharipov2011}
the Maxwell model with a single AC is inadequate to 
capture all the scattering behavior of the reflected molecules.
Alternatively, the classical Cercignani-Lampis (CL) scattering
kernel \cite{CLL} has two ACs individually measuring
the exchange of tangential momentum and normal energy.
The CL boundary conditions (BCs) are more realistic
and agree well with some experimental and MD data.
\cite{Su2019,And2019}

Due to the extreme importance of slip and jump 
coefficients, many theoretical and numerical efforts have
been paid to study these issues, especially for general
gas-surface interactions beyond the Maxwell model.
Slip and jump coefficients can be obtained
from half-space layer equations \cite{Bardos2006,Sone2007,Bern2008}
that depict the boundary behavior of the rarefied gas
flows. The corresponding explicit expressions in terms
of a set ACs have been given by the variational principle
\cite{CL1989,Klinc1972,Loren2011,NN2020} and the 
low-order moment-type methods.  
These moment methods usually assume the velocity distribution
function to be a Maxwellian multiplied by low-order polynomials. 
For the Maxwell diffuse-specular BCs, there are some famous
moment-type methods such as the Maxwell method, \cite{Maxwell}
the Loyalka method, \cite{Loyalka1967} and the half-range moment
method. \cite{Gross1957} Similar ideas have been applied 
to the CL BCs. \cite{Struch2013,Zhang2021}
There are direct numerical methods to calculate slip and
jump coefficients. In contrast to the Maxwell model,
\cite{Sharipov2011} fewer data are available for the CL BCs.
Siewert has developed an analytical discrete-ordinates method 
to numerically solve the layer problems based on the linearized 
Boltzmann equation (LBE). \cite{Siewert2003c} Note that the above work
is all restricted to the simplified collision models or the
hard-sphere potential. For general intermolecular potentials,
the LBE is directly solved by a synthetic 
iteration scheme and fast spectral method recently
to predict slip and jump coefficients. \cite{Su2019,Su2020,Su2022} 

The paper will focus on the Grad moment method
\cite{Grad1949} based on the arbitrary order Hermite
expansion of the velocity distribution function. The Grad-type
moment model with 13 moments, \cite{Tah2009,Tah2010} 26 moments,
\cite{Gu2010,Gu2014} and an arbitrary number of 
moments \cite{2008Linear,Lijun2017,Yang2022a}
have been developed to study the slip and jump. 
It's promising that the accuracy of the moment model can 
improve when we enlarge the number of used moments. However,
the related theory is mainly on the Maxwell BCs, and the
numerical results are mostly reported for simplified 
collision models. Here, we
develop the arbitrary order moment method to model the
layer problems with general BCs. In particular, we achieve
the numerical solutions and explicit formulae of slip and
jump coefficients for the CL BCs and inverse-power-law
(IPL) intermolecular potentials.

There are three distinct difficulties when we extend the 
moment methods to general cases. Firstly, although the 
general mathematical formulation of the layer problems based on 
the LBE is well-known, \cite{Williams2001,Siewert2003c} it's hard
to find a clear and rigorous definition of the layer problems for 
the moment method with arbitrary order and IPL potentials. Secondly, 
the solvability of the layer equations is not for granted.
The general well-posed theory of the boundary value problem
for the moment equations are studied in 
Ref.\onlinecite{Yang2022a,Yang2022b}. 
Similar results about the discrete Boltzmann equation are given
in Ref.\onlinecite{Bern2008,Bern2010b,Bern2010}. We need to verify the 
solvability conditions for the moment equations with general BCs.
Thirdly, the numerical treatment of the layer equations with 
general BCs is not trivial. The calculation of the IPL potentials
has been recently considered in Ref.\onlinecite{Wang2019}.
The Hermite expansion of the CL scattering kernel is 
mainly a two-dimensional half-space integral, which is
a combination of the modified Bessel function, Hermite 
polynomials, exponentials, and powers. The reckless numerical
integration may bring considerable errors and prevent us 
from finding explicit expressions of slip and jump coefficients
in terms of the ACs.

The paper is devoted to all of the above issues.
We define the layer problems based on the 
moment method with arbitrary order according to 
the spirit of the Chapman-Enskog
expansion. The model can depict all classical half-space
problems with IPL potentials. 
To ensure the solvability, we make a careful choice of
the test functions and the simple boundary stabilization with
a rank-one modification. Briefly speaking,
the rank-one modification removes the eigenvalue one
arising from the normalization property of the scattering
kernel, which leads to the positive definiteness of the
symmetric coefficient matrix. To obtain explicit formulae,
we give closed-form 
integral representations involving the CL scattering kernel
and Hermite polynomials by a recursion relation. 
These formulae explain the influence of the two ACs
in the CL kernel and the different intermolecular potentials.

The rest of this paper is organized as follows. In Section
\ref{sec:2}, we briefly review the layer equations and
derive the general BCs for them.
In Section \ref{sec:3}, we verify the solvability
of the layer equations. In Section \ref{subsec:32},
we focus on the calculation of the CL scattering kernel.
In Section \ref{sec:4}, we consider the specific 
half-space problems. We give highly accurate numerical
results and explicit formulae about slip and jump 
coefficients. The paper ends with conclusions.


	\section{The Moment Model}
	\label{sec:2}
	\subsection{Review on the layer equations}
For single-species monatomic gases, the half-space equations 
for the moment model read as \cite{Lijun2017,Yang2022a} 
\begin{eqnarray}
	\label{eq:KL}
	\+A_2\od{\+w}{y} &=& -\+Q\+w,\ \+w=\+w(y),\ y\in[0,+\infty),\\
	\+w(\infty) &=& \+0, \notag
\end{eqnarray}
where $\+w=\+w(y)\in\bbR^N$ is the moment variable in
the Knudsen layer and $\+A_2$ as well as $\+Q$ are constant matrices.
The argument $y$ is the stretched coordinate normal to the
boundary. This class of half-space equations can depict the
boundary behavior of the rarefied gas flows. 
\cite{Bardos2006,Sone2007,Bern2008}

The moment variable is from the Hermite series of
the velocity distribution function.
Let $\omega = \omega(\+\xi)$ be the global Maxwellian
\begin{equation}
	\label{eq:omega1}
	\omega(\+\xi) = (2\pi)^{-3/2}\exp\left(-|\+\xi|^2/2\right),
\end{equation}
where $\+\xi=(\xi_1,\xi_2,\xi_3)\in\bbR^3$ is 
the microscopic velocity of particles.
The corresponding Hermite polynomials 
$\phi_{\+\alpha}=\phi_{\+\alpha}(\+\xi)$ are defined \cite{Grad1949N}
by the recursion relation
\begin{equation}
	\label{eq:defH}
	\xi_d\phi_{\+\alpha} = \sqrt{\alpha_d}\phi_{\+\alpha-\+e_d}
	+\sqrt{\alpha_d+1}\phi_{\+\alpha+\+e_d},\
	\+\alpha=(\alpha_1,\alpha_2,\alpha_3)\in\bbN^3,
\end{equation}
where $\phi_{\+0}=1$ and $\phi_{\+e_i}=\xi_i$, with
$\+e_i\in\bbN^3$ only the $i$-th element being one. 
As a convention, we regard $\phi_{\+\alpha}$ as zero
if any component of $\+\alpha$ is negative.
Then we have the orthogonality
\begin{equation}
	\ang{\omega\phi_{\+\alpha}\phi_{\+\beta}} = \delta_{\+\alpha,
	\+\beta},\quad  
	\ang{\cdot}\triangleq\int_{\bbR^3}\!\!\cdot\,\md\+\xi.
\end{equation}
For any given integer $M\geq 3$, we let
\[
	\mathbb{I}_M=\{\+\alpha\in\bbN^3,\ |\+\alpha|=\alpha_1+
\alpha_2+\alpha_3\leq M\},\quad N=\#\mathbb{I}_M.
\]
In virtue of the ordering of multi-indices
(represented by the square brackets), we call the moment 
variable $\+w$ ``induced'' from $\mathbb{I}_M$ with
\[
	\+w[\+\alpha] \triangleq w_{\+\alpha} = w_{\+\alpha}(y).
\]
The above notations are rigorously specified in Appendix \ref{app:A}.
Then $M$ is called the moment order and the perturbation of 
the velocity distribution function in the Knudsen layer can
be approximated by the Hermite series
\begin{equation}
	\label{eq:ans1}
	f = f(y,\+\xi) = \omega(\+\xi)\sum_{\+\alpha\in\mathbb{I}_M}
	w_{\+\alpha}(y)
	\phi_{\+\alpha}(\+\xi).
\end{equation}
The formula \eqref{eq:ans1} relates the moment variable with
physical quantities in the Knudsen layer. For example,
the density $\rho$, the temperature $\theta$, the macro
velocity $\+u=(u_1,u_2,u_3)$, the stress tensor $\sigma_{ij}$,
and the pressure $p$ are defined as
\begin{equation}
	\label{eq:macro_v}
	\rho=\ang{f}, u_i=\ang{\xi_i f}, \theta=
	\ang{\left(\frac{|\+\xi|^2}{3}-1\right)f}, p=\rho+\theta,
	\sigma_{ij}+p\delta_{ij}=\ang{\xi_i\xi_j f}.
\end{equation}

The matrices $\+A_2\in\bbR^{N\times N}$ and 
$\+Q\in\bbR^{N\times N}$ are induced 
(see Appendix \ref{app:A}) from
$\mathbb{I}_M\times\mathbb{I}_M$ with
\begin{eqnarray}
	\label{eq:A}
	\+A_2[{\+\alpha},{\+\beta}]&=&
	\ang{\xi_2\omega\phi_{\+\alpha}\phi_{\+\beta}} = 
\sqrt{\alpha_2}\delta_{\+\beta,\+\alpha-\+e_2}+
\sqrt{\alpha_2+1}\delta_{\+\beta,\+\alpha+\+e_2},\\
	\label{eq:Q}
	\+Q[{\+\alpha},{\+\beta}]&=&	
	\ang{\omega\mathcal{L}(\phi_{\+\beta})\phi_{\+\alpha}}.	
\end{eqnarray}
Here the operator $\mathcal{L}$ is the linearized Boltzmann
operator defined by
\[
	\mathcal{L}(\phi) = -\frac{1}{B_0}
		\omega^{-1}(\mathcal{Q}(\omega,\omega\phi)+
		\mathcal{Q}(\omega\phi,\omega)),
\]
where $B_0$ is a constant representing the average
collision frequency, and $\mathcal{Q}$ is the Boltzmann
collision operator defined as 
\begin{equation}
	\notag
	\mathcal{Q}(g,h) = \frac{1}{2}
	\int_{\bbR^3}\!\!\int_{\bbS^{2}}\!
   	\!(g'h_*'+g_*'h'-gh_*-g_*h)
	B(|\+\xi-\+\xi_*|,\+\Theta)\,\mathrm{d}\+\Theta
	\mathrm{d}\+\xi_*.
\end{equation}
Note that we write $g_*=g(t,\+x,\+\xi_*),\ g'=g(t,\+x,\+\xi'),\ g_*'=
g(t,\+x,\+\xi_*'),$ etc., for short. The pre-collisional
velocities $\+\xi'$ and $\+\xi_*'$ are
\[
\+\xi' = \frac{\+\xi+\+\xi_*}{2}+\frac{|\+\xi-\+\xi_*|}{2}\+\Theta,
\quad 	
\+\xi_*' = \frac{\+\xi+\+\xi_*}{2}-\frac{|\+\xi-\+\xi_*|}{2}\+\Theta,
\quad \+\Theta\in\bbS^2.
\]
Let $\+g=\+\xi-\+\xi_*$. The 
nonnegative function $B(|\+g|,\+\Theta)$ is called the collision kernel.

In this paper, we consider the inverse-power-law (IPL) model.
Now the collision kernel has the form \cite{Wang2019}
\[
B(|\+g|,\+\Theta) = |\+g|^{\frac{\eta-5}{\eta-1}}W_0
\left|\od{W_0}{\varphi}\right|,
\]
where $\varphi$ is the angle satisfying
$\cos\varphi=\+g\cdot\+\Theta/|\+g|$ and $\eta>3$ is the index
in the inverse-power potential. $3<\eta<5$ is called the
soft potential and $\eta>5$ is called the hard potential.
When $\eta=5$, it is the case
of Maxwell molecules. When $\eta=+\infty$, the model
can be regarded as the hard-sphere (HS) model. The dimensionless
impact parameter $W_0$ is given by
\[
\varphi = \pi-2\int_0^{W_1}\left(1-W^2-\frac{2}{\eta-1}\left(
			\frac{W}{W_0}\right)^{\eta-1}\right)^{-1/2}\,\mathrm{d}W,
\]
with $W_1$ a positive real number satisfying
\[
1-W_1^2-\frac{2}{\eta-1}\left(\frac{W_1}{W_0}\right)^{\eta-1}=0.	
\]
For the IPL intermolecular potentials
defined above,
it is classical \cite{CC1989} that
$\+Q\geq 0$.

	\subsection{General boundary conditions}
Assume the boundary is fixed at the plane $\{\+x\in\bbR^3,\ x_2=0\}.$
In the kinetic theory, the general scattering BCs for the 
velocity distribution function $F(\+\xi)$ read as
\begin{equation}
	\label{eq:gBC}
	\xi_2 F(\+\xi) = \int_{\xi_2'<0}
	\!\!|\xi_2'|R(\+\xi'\rightarrow\+\xi)F(\+\xi')\,\mathrm{d}
	\+\xi',\ \xi_2>0.
\end{equation}
The scattering kernel $R(\+\xi'\rightarrow\+\xi)$
describes the probability that an incident molecule with the
velocity $\+\xi'$ is scattered to have
the reflected velocity $\+\xi$ lying in the range $\mathrm{d}\+\xi$.
In principle, the scattering kernel is defined only for $\+\xi$ and
$\+\xi'$ with $\xi_2>0$ as well as $\xi_2'<0$. The kernel
$R(\+\xi'\rightarrow\+\xi)$
must satisfy some general conditions \cite{CLL} such as
\begin{itemize}
\item Nonnegativity.
$R(\+\xi'\rightarrow\+\xi)\geq 0,$ for all $\+\xi,\+\xi'.$
\item Normalization.
	\begin{equation}
		\label{eq:cl_nor}
		\displaystyle\int_{\xi_2>0}\!\!R(\+\xi'\rightarrow\+\xi)
\,\mathrm{d}\+\xi=1,\ \forall \+\xi'.
	\end{equation}
\item Detailed balance.
	\begin{equation}
		\label{eq:cl_deb}
		|\xi_2'|\omega(\+\xi')R(\+\xi'\rightarrow\+\xi)=
\xi_2\omega(\+\xi)R(-\+\xi\rightarrow -\+\xi').
	\end{equation}
\end{itemize}

The linearized Maxwell diffuse-specular scattering
kernel \cite{Maxwell} reads as
\begin{equation}
R(\+\xi'\rightarrow\+\xi) = 
R_M(\+\xi'\rightarrow\+\xi) =
\chi\frac{\xi_2}{2\pi}
\exp\left(-\frac{|\+\xi|^2}{2}\right) +
(1-\chi)\delta(\+\xi'-\+\xi^*),
\end{equation}
where $\chi\in[0,1]$ is the tangetial momentum accommodation
coefficient (TMAC) and $\+\xi^*=(\xi_1,-\xi_2,\xi_3)$.
The linearized Cercignani-Lampis (CL) scattering kernel
\cite{CLL} reads as 
\begin{eqnarray} &&
R(\+\xi'\rightarrow\+\xi) = 
R_{CL}(\+\xi'\rightarrow\+\xi) = \\
&=& \notag
	\frac{\xi_2}{2\pi\alpha_n\alpha_t(2-\alpha_t)}
	I_0\left(\frac{\sqrt{1-\alpha_n}}{\alpha_n}\xi_2\xi_2'\right)
	\exp\left(-\frac{\xi_2^2+(1-\alpha_n)\xi_2'^2}{2\alpha_n}
			-\frac{|\+\xi_t-(1-\alpha_t)\+\xi_t'|^2}
			{2\alpha_t(2-\alpha_t)}\right),
\end{eqnarray}
where $0<\alpha_n\leq 1$ and $0<\alpha_t<2$ are 
two accommodation coefficients (ACs).
Here $\+\xi_t=(\xi_1,\xi_3)$ and 
the modified Bessel function 
\[I_0(x)=\displaystyle\frac{1}{2\pi}
\int_0^{2\pi}\!\!\exp(x\cos(\varphi))\,\mathrm{d}\varphi.\]
When $\alpha_n=\alpha_t=1$, the CL BCs
turn to the fully diffuse case, samely as
$\chi=1$ in the Maxwell BCs.
As $\alpha_n\rightarrow 0$
and $\alpha_t\rightarrow 0$, the specular BCs are recovered.
When $\alpha_n\rightarrow 0$ and $\alpha_t\rightarrow 2$,
we have the bounce-back condition. The CL BCs are far from
general \cite{Cow1974} but simple and accurate enough to 
be used in practice.

We follow Grad's framework \cite{Grad1949} to construct the
BCs for moment equations \eqref{eq:KL} by testing the 
kinetic BCs \eqref{eq:gBC}
with even polynomials (about $\xi_2$). 
Choosing the Hermite polynomials $\phi_{\+\alpha}$ with even $\alpha_2$
as test functions gives
\begin{eqnarray}
	\label{eq:gBC2}
	\int_{\xi_2>0}\!\!\xi_2 F(\+\xi)\phi_{\+\alpha}(\+\xi)
	\,\mathrm{d}\+\xi = \int_{\xi_2>0}\!\!\int_{\xi_2'<0}
	\!\!|\xi_2'|R(\+\xi'\rightarrow\+\xi)F(\+\xi')\,\mathrm{d}
	\+\xi'\phi_{\+\alpha}(\+\xi)\,\mathrm{d}\+\xi.
\end{eqnarray}
We write the velocity distribution function as
\[
	F(\+\xi) = f(0,\+\xi) + \bar{F}(\+\xi),
\]
where $f$ is the perturbation defined in \eqref{eq:ans1} and
$\bar{F}(\+\xi)$ is regarded as the given velocity
distribution function of the bulk flow. 
Analogously as \eqref{eq:ans1}, we assume
\[
	\bar{F}(\+\xi) = \omega(\+\xi)\sum_{\+\alpha\in\mathbb{I}_M}
	\bar{w}_{\+\alpha}\phi_{\+\alpha}(\+\xi).
\]
Plugging \eqref{eq:ans1} and the above formula into \eqref{eq:gBC2},
denoting
\[ K(\+\xi'\rightarrow\+\xi) = \frac{1}{\xi_2}R(\+\xi'\rightarrow
\+\xi)\omega(\+\xi'),\ \xi_2>0,\xi_2'<0,\]
we have the BCs for the moment variable $\+w$:
\begin{eqnarray}\notag
	&&\sum_{|\+\beta|\leq M} (w_{\+\beta}(0)+\bar{w}_{\+\beta})
	\int_{\xi_2>0}\!\!\xi_2 \omega(\+\xi)
	\phi_{\+\beta}(\+\xi)\phi_{\+\alpha}(\+\xi)
	\,\mathrm{d}\+\xi \\ \label{eq:gBC3}
	&=& \sum_{|\+\beta|\leq M} (w_{\+\beta}(0)+\bar{w}_{\+\beta})
	\int_{\xi_2>0}\!\!\int_{\xi_2'<0}
	\!\!\xi_2|\xi_2'| 
	K(\+\xi'\rightarrow\+\xi)\phi_{\+\beta}(\+\xi')
	\phi_{\+\alpha}(\+\xi)\,\mathrm{d}\+\xi'\mathrm{d}\+\xi,
\end{eqnarray}
where $\alpha_2$ is even and $\bar{w}_{\+\alpha}$ are regarded
as given by the bulk flow.

\begin{remark}
	The detailed balance \eqref{eq:cl_deb}
	of $R(\+\xi'\rightarrow\+\xi)$ ensures the reciprocity
	of $K(\+\xi'\rightarrow\+\xi)$, i.e.,
	\[K(\+\xi'\rightarrow\+\xi)=K(-\+\xi\rightarrow -\+\xi').\]
\end{remark}

For simplicity, we write the half-space integrals as 
\begin{eqnarray}\label{eq:defS}
	S(\+\alpha,\+\beta) = \int_{\xi_2>0}\!\!\xi_2 \omega(\+\xi)
	\phi_{\+\beta}(\+\xi)\phi_{\+\alpha}(\+\xi)
	\,\mathrm{d}\+\xi.
\end{eqnarray}
In order to analyze, we make an even continuation
such that 
$K(\+\xi'\rightarrow\+\xi)=K(-\+\xi'\rightarrow\+\xi)$
for $\+\xi'\in\bbR^3.$ Then we define
\begin{eqnarray}
	\label{eq:defR}
	R(\+\alpha,\+\beta) &=& \int_{\xi_2>0}\!\!\xi_2
	\left(\int_{\bbR^3}\!\! K(\+\xi'\rightarrow\+\xi)
	\phi_{\+\beta}(\+\xi')\,\mathrm{d}\+\xi'\right)
	\phi_{\+\alpha}(\+\xi)\,\mathrm{d}\+\xi. 
\end{eqnarray}
Since the whole space integral about odd functions are zero,
we have $R(\+\alpha,\+\beta)=0$ when $\beta_2$ is odd.
If all the mentioned integrals exist and the Fubini theorem holds,
we have
	\begin{eqnarray*}
		&&\int_{\xi_2>0}\!\!\int_{\xi_2'<0}
	\!\!\xi_2|\xi_2'| 
	K(\+\xi'\rightarrow\+\xi)\phi_{\+\beta}(\+\xi')
		\phi_{\+\alpha}(\+\xi)\,\mathrm{d}\+\xi'\mathrm{d}\+\xi \\
		&=& \int_{\xi_2'<0}\!\!|\xi_2'|\left(\int_{\xi_2>0}
	\!\!\xi_2
	K(\+\xi'\rightarrow\+\xi)
		\phi_{\+\alpha}(\+\xi)\,\mathrm{d}\+\xi\right)
		\,\mathrm{d}\+\xi'\phi_{\+\beta}(\+\xi') \\
		&=& \int_{\xi_2'<0}\!\!|\xi_2'|\sum_{\+\gamma\in\bbN^3}
		\left(\int_{\bbR^3}\!\!\left(\int_{\xi_2>0}
	\!\!\xi_2 K(\+\xi'\rightarrow\+\xi)
		\phi_{\+\alpha}(\+\xi)\,\mathrm{d}\+\xi\right)
		\phi_{\+\gamma}(\+\xi')\,\md\+\xi'\right)
		\phi_{\+\beta}(\+\xi')\phi_{\+\gamma}(\+\xi')\omega(\+\xi') 
		\,\mathrm{d}\+\xi' \\
		&=& \sum_{\+\gamma\in\bbN^3}
		R(\+\alpha,\+\gamma)
		\int_{\xi_2'<0}\!\!|\xi_2'|
		\phi_{\+\beta}(\+\xi')\phi_{\+\gamma}(\+\xi')\omega(\+\xi') 
		\,\mathrm{d}\+\xi' 	\\
		&=& \sum_{\+\gamma\in\bbN^3}R(\+\alpha,\+\gamma)
		S(\+\gamma,\+\beta)(-1)^{\beta_2+\gamma_2}.
	\end{eqnarray*}
Meanwhile, the normalization property \eqref{eq:cl_nor} gives
$R(\+0,\+\beta)=\delta_{\+\beta,\+0}.$

Now we give the detailed choice of the test functions.
We split the multi-indices into even and odd subsets, i.e.,
\[
	\mathbb{I}_{M,e}=\{\+\alpha\in\bbN^3:\ |\+\alpha|\leq M,\
			\alpha_2 \text{\ even}\},\quad
	\mathbb{I}_{M,o}=\{\+\alpha\in\bbN^3:\ |\+\alpha|\leq M,\
	\alpha_2 \text{\ odd}\}.
\]
Let $m=\#\mathbb{I}_{M,e}$ and $n=\#\mathbb{I}_{M,o}$.
We immediately have $m+n=N$ with $m\geq n.$ So as in 
Grad's framework, \cite{Grad1949} a natural choice of
the test functions is $\phi_{\+\alpha}$ with
$\+\alpha\in\mathbb{I}_{M-1,e}$. The choice exactly 
gives $n=\#\mathbb{I}_{M-1,e}$ BCs. Let 
$\+w=[\+w_e^T,\+w_o^T]^T$ and 
$\bar{\+w}=[\bar{\+w}_e^T,\bar{\+w}_o^T]^T.$ Here  
$\+w_e\in\bbR^m$ and $\bar{\+w}_e\in\bbR^m$ are induced (see
Appendix \ref{app:A}) from $\mathbb{I}_{M,e}$ and $\+w_o\in\bbR^n,
\bar{\+w}_o\in\bbR^n$ are induced from $\mathbb{I}_{M,o}$.
	For any positive integers $I$ and $J$, assume that
	the matrix $\+M_{I,J}$ is induced from $\mathbb{I}_{I,e}
	\times\mathbb{I}_{J,o}$, and $\+S_{I,J},\ \+R_{I,J}$ are
	induced from $\mathbb{I}_{I,e}\times\mathbb{I}_{J,e}$, with
	the entries 
\begin{eqnarray}
	\label{eq:defMM}
	\+M_{I,J}[{\+\alpha},{\+\beta}] =
	2S(\+\alpha,\+\beta),\quad  
	\+S_{I,J}[{\+\alpha},{\+\beta}] =
	S(\+\alpha,\+\beta),\quad
	\+R_{I,J}[{\+\alpha},{\+\beta}] =
	R(\+\alpha,\+\beta).
\end{eqnarray}
Then from \eqref{eq:gBC3}, the BCs of the moment 
equations \eqref{eq:KL} can 
write in the even-odd parity form
\begin{eqnarray}\notag &&
	(\+S_{M-1,M}-\+R_{M-1,J}\+S_{J,M})(\+w_e(0)+\bar{\+w}_e)
	\\ &+& \frac{1}{2}(\+M_{M-1,M}+\+R_{M-1,J}\+M_{J,M})
	(\+w_o(0)+\bar{\+w}_o)=\+0, \label{eq:gBC4}
\end{eqnarray}
where $\+M_{M-1,M}\in\bbR^{n\times n}$,
$\+S_{M-1,M}\in\bbR^{n\times m}$,
and $J$ is any positive integer.

\begin{remark}
	In general, the BCs in \eqref{eq:gBC4} is equivalent to the
	corresponding ones in \eqref{eq:gBC3} when $J=\infty$.
	But for the CL scattering kernel, we will show below
	(in Theorem \ref{thm:scll}) that  
	\[R(\+\alpha,\+\beta)=0,\ \text{\ when}\ \beta_2>\alpha_2.\]
	So $J = M-1$ is enough in the CL case.
\end{remark}

	\section{Solvability of the Layer Equations}
	\label{sec:3}
	To ensure the unique solvability of the layer equations \eqref{eq:KL}:
\[
\+A_2\od{\+w}{y} = -\+Q\+w,\ \+w(\infty)=\+0,\ y\in[0,+\infty),
\]
the proper number of BCs at $y=0$ relies on not only
the eigenvalues of the boundary matrix, i.e., $-\+A_2$ 
here, but also the null space structure of $\+Q.$ 
As shown in Ref.\onlinecite{Grad1949,cai2011}, the $N\times N$ matrix
$\+A_2$ has $n$ positive eigenvalues,
$n$ negative eigenvalues and $m-n$ zero eigenvalues.
It's classical \cite{CC1989} that the linearized
Boltzmann operator has a five-dimensional null space.
Accordingly, we have \cite{Yang2022a}
$\mathrm{Null}(\+Q)=\mathrm{span}\{\+\varphi_0,
\+\varphi_1,\+\varphi_2,\+\varphi_3,\+\varphi_4\}$, where
$\+\varphi_i$ are vectors induced from $\mathbb{I}_M$ with
the non-zero entries
\begin{eqnarray*}
	\+\varphi_0[{\+0}] = 1,\quad
	\+\varphi_i[{\+e_i}] = 1,\ 1\leq i\leq 3,\quad
	\+\varphi_{4}[{2\+e_d}] = \sqrt{3}/3,\ 1\leq d\leq 3.
\end{eqnarray*}
From \eqref{eq:macro_v}, 
we have $\+\varphi_0^T\+w=\rho,\ \+\varphi_i^T\+w=u_i,
1\leq i\leq 3,$ and $\+\varphi_4^T\+w
=\displaystyle\frac{\sqrt{6}}{2}\theta$. 
Then Theorem 1 in Ref.\onlinecite{Yang2022a} shows that
the layer equations 
\eqref{eq:KL} need $n-4$ BCs at $y=0$ to determine a
unique $\+w(y)$. 

The above result is different from the well-posed theory of the
initial boundary value
problem (IBVP), where $n$ BCs are needed. \cite{Hil2013}
To agree with the IBVP, we also impose $n$ BCs for the layer
equations \eqref{eq:KL}. Therefore, the given values $\bar{\+w}_e$ and
$\bar{\+w}_o$ in the BCs \eqref{eq:gBC4} should not be arbitrary
but satisfy $4$ additional conditions. Theorem 2 in 
Ref.\onlinecite{Yang2022a}
shows that a special structure of the BCs would ensure the solvability,
i.e.,

\begin{theorem}
	\label{thm:31}
	If we equip the layer equations \eqref{eq:KL} with
	the following $n$ BCs at $y=0$:
	\begin{eqnarray}\label{eq:stru}
	\+M_{M,M}^T(\+w_e(0)+\bar{\+w}_e) + \+H(\+w_o(0)+\bar{\+w}_o) = \+0,	
	\end{eqnarray}
	where $\+M_{M,M}\in\bbR^{m\times n}$ is defined in \eqref{eq:defMM}
	and $\+H\in\bbR^{n\times n}$ is any symmetric positive definite
	matrix, then there
	exists a unique solution of $\+w(y)$ and $\+\varphi_i^T\bar{\+w},
	i=0,1,3,4,$ when other components of
	$\bar{\+w}$ are arbitrarily given.
	In particular, the solution would give $\+\varphi_2^T\+w(y)
	=u_2(y)=0.$
\end{theorem}

Unfortunately, \eqref{eq:gBC4} is not in the form of
\eqref{eq:stru}. For the Maxwell diffuse-specular case,
\eqref{eq:gBC4} is shown unstable when $m>n$ for the 
non-homogeneous layer equations. \cite{Yang2022b}
The reason lies in that the linear space determined
by the BCs does not contain the null space of the boundary matrix.
\cite{Osher1975,Sarna2018} In Ref.\onlinecite{Sarna2018,Yang2022a},
the Maxwell BCs are modified to agree with the structure in
\eqref{eq:stru}. A basic tool in Ref.\onlinecite{Yang2022a} is
the following lemma:

\begin{lemma}
	\label{lem:01}
	For any positive integer $I$ and $J$, the matrix $\+S_{I,I}$ is 
	symmetric positive definite.
	The matrix $\+M_{I,J}$ is lower triangular. When $I\geq J$, 
	the matrix $\+M_{I,J}$ is of full column rank.
\end{lemma}

To apply the solvability theorem \ref{thm:31} for general
gas-surface scattering kernels, we follow Ref.\onlinecite{Sarna2018}'s
idea to modify \eqref{eq:gBC4} as minor as possible such
that the modified BCs have the form as \eqref{eq:stru}.
We first let $J=M-1$ in \eqref{eq:gBC4} and change all
$\+S_{M-1,M}$ to $\+S_{M-1,M-1}$. This way discards the
highest order moment variables when $m>n$. Then we try 
to inverse the coefficient matrix before
$\+w_e+\bar{\+w}_e$ to meet the form in \eqref{eq:stru}. 
However, due to the normalization property \eqref{eq:cl_nor}
of $R(\+\xi'\rightarrow\+\xi)$, the corresponding coefficient matrix
(i.e., $\+I-\+R_{M-1,M-1}$ in \eqref{eq:gBC5})
is irreversible. We should make a clever rank-one modification 
to ensure the invertibility (of $\+I-\hat{\+R}$ in \eqref{eq:gBC5}).
In a word, the modified BCs read as
\begin{gather}
	\label{eq:gBC5}
	\+M_{M,M}^T(\+w_e(0)+\bar{\+w}_e) + \+H(\+w_o(0)+\bar{\+w}_o) 
	= \+0, \\ \notag
	  \+H = \frac{1}{2}\+M_{M-1,M}^T\left\{
		  \+S_{M-1,M-1}^{-1}\left(\+I-\hat{\+R}
			\right)^{-1} 
	\left(\+I+\hat{\+R}\right)\right\}\+M_{M-1,M},\\
\notag \hat{\+R} = \+R_{M-1,M-1} - \frac{\+S_{M-1,M-1}
	\+e\+e^T}{\+e^T\+S_{M-1,M-1}\+e}
	\in\bbR^{n\times n},
\end{gather}
where $\+I$ is the $n$-th order identity matrix and 
$\+e\in\bbR^n$ is a unit vector induced from $\mathbb{I}_{M-1,e}$
with $\+e[{\+0}]=1.$ We claim that the modified BCs 
\eqref{eq:gBC5} has the structure as \eqref{eq:stru}:

\begin{lemma}
	\label{lem:02}
	If $\hat{\+R}$ satisfies $\rho(\hat{\+R})=k<1$ and
	\begin{eqnarray}
		\label{eq:DB}
		\hat{\+R}\+S_{M-1,M-1}=\+S_{M-1,M-1}\hat{\+R}^T,
	\end{eqnarray}
	then the matrix $\+H$ defined in \eqref{eq:gBC5} is
	symmetric positive definite. Here $\rho(\hat{\+R})$
	is the spectral radius defined as the maximum modulus
	eigenvalues of $\hat{\+R}$.
\end{lemma}

\begin{proof}
	By Lemma \ref{lem:01}, we have $\+S_{M-1,M-1}>0$. 
	Since $\rho(\hat{\+R})<1$, the matrix
	$\+I-\hat{\+R}$ is invertible.
	Let $\+X=(\+I-\hat{\+R})^{-1}\+M_{M-1,M}
	\in\bbR^{n\times n}$. Utilizing \eqref{eq:DB}, we have
	\[
		\+S_{M-1,M-1}^{-1}-\hat{\+R}^T\+S_{M-1,M-1}^{-1}\hat{\+R}
		= (\+I-\hat{\+R}^T)\+S_{M-1,M-1}^{-1}(\+I+\hat{\+R}).
	\]
	Noting that $(\+I-\hat{\+R})^{-1}(\+I+\hat{\+R})=
	(\+I+\hat{\+R})(\+I-\hat{\+R})^{-1}$,
	from \eqref{eq:gBC5} we have
	\[
		\+H = \frac{1}{2}\+X^T
		\left(\+S_{M-1,M-1}^{-1}-\hat{\+R}^T\+S_{M-1,M-1}^{-1}
		\hat{\+R}\right)\+X.
	\]
	Although $\hat{\+R}$ may not be symmetric, the condition
	\eqref{eq:DB} gives
	\[ \+S_{M-1,M-1}^{-1/2}\hat{\+R}\+S_{M-1,M-1}^{1/2}
	= \left(\+S_{M-1,M-1}^{-1/2}
		\hat{\+R}\+S_{M-1,M-1}^{1/2}\right)^T.\]
	Since similar matrices have the same eigenvalues,
	we have $\|\+S_{M-1,M-1}^{-1/2}
		\hat{\+R}\+S_{M-1,M-1}^{1/2}\|_2 = k$.
	So for any $\+y\in\bbR^n$, we have 
	\begin{eqnarray*}&&
		\+y^T\left(\+S_{M-1,M-1}^{-1}-\hat{\+R}^T\+S_{M-1,M-1}^{-1}
		\hat{\+R}\right)\+y \\&=& \+y^T\+S_{M-1,M-1}^{-1/2}
		\left(\+I-
		\left(\+S_{M-1,M-1}^{1/2}\hat{\+R}^T\+S_{M-1,M-1}^{-1}
		\hat{\+R}\+S_{M-1,M-1}^{1/2}\right)
		\right)\+S_{M-1,M-1}^{-1/2}\+y \\
		&\geq& \left(1-\|\+S_{M-1,M-1}^{-1/2}
		\hat{\+R}\+S_{M-1,M-1}^{1/2}\|_2^2\right)
		\+y^T\+S_{M-1,M-1}^{-1}\+y \\
		&\geq& \left(1-k^2\right)
		\+y^T\+S_{M-1,M-1}^{-1}\+y.
	\end{eqnarray*}
	Thus, under the lemma's conditions, the matrix
	$\+S_{M-1,M-1}^{-1}-\hat{\+R}^T\+S_{M-1,M-1}^{-1}\hat{\+R}$
	is symmetric positive definite.
	Combined with Lemma \ref{lem:01}, the matrix $\+H$
	in \eqref{eq:gBC5} is also symmetric positive definite.
\end{proof}

\begin{remark}
	\label{rem:01}
	The matrix $\hat{\+R}$ in \eqref{eq:gBC5} is constructed from
	the belief that the solutions of \eqref{eq:KL} with 
	\eqref{eq:gBC4} should be recovered from the solutions of
	\eqref{eq:KL} with the modified BCs \eqref{eq:gBC5} under
	some conditions. When $J=M-1$ and $m=n$ (the condition
	$m=n$ is possible when $\+\alpha\in\bbN$ and $M$ is odd),
	we have $\+M_{M-1,M}=\+M_{M,M}$ invertible and
	$\+S_{M-1,M-1}=\+S_{M-1,M}$.
	So the modified BCs \eqref{eq:gBC5} are equivalent to
\begin{eqnarray}\notag
	(\+S_{M-1,M}-\hat{\+R}\+S_{M-1,M})(\+w_e+\bar{\+w}_e)
	+ \frac{1}{2}(\+M_{M-1,M}+
	\hat{\+R}\+M_{M-1,M})
	(\+w_o+\bar{\+w}_o)=\+0. \notag
\end{eqnarray}
Since $\+M_{M-1,M}$ is lower triangular, when
	$\+w_o[{\+e_1}]=\bar{\+w}_o[{\+e_1}]=0$ and $\hat{\+R}$
	is defined as \eqref{eq:gBC5}, we have
	\[(\hat{\+R}-\+R_{M-1,M-1})\+M_{M-1,M}(\+w_o+\bar{\+w}_o)=\+0.\]
	Compared with \eqref{eq:gBC4}, we try to find $\rho^w\in\bbR$
	such that
	\[
		(\+S_{M-1,M}-\hat{\+R}\+S_{M-1,M})(\+w_e+\rho^w\+e)
		= 
		(\+S_{M-1,M}-\+R_{M-1,M-1}\+S_{M-1,M})\+w_e.
	\]
	If the detailed balance condition \eqref{eq:DB} and
	the normalization condition $\+R_{M-1,M-1}^T\+e=\+e$ hold,
	after some manipulation we find that $\hat{\+R}^T\+e=\+0$ and 
	\[
		\rho^w = \+e^T\+S_{M-1,M}\+w_e/\+e^T\+S_{M-1,M}\+e.	
	\]
	Hence, the recovery is easy to obtain.
\end{remark}

For general gas-surface scattering kernels, if the
conditions in Lemma \ref{lem:02} hold, then Theorem
\ref{thm:31} ensures the solvability of the layer equations
\eqref{eq:KL} with the modified BCs \eqref{eq:gBC5}. We 
can see that the condition \eqref{eq:DB} is naturally 
from the detailed balance property \eqref{eq:cl_deb}
of $R(\+\xi'\rightarrow\+\xi)$. We do not try to check
the condition $\rho(\hat{\+R})<1$ for general BCs.
But as will be shown below, the matrix $\hat{\+R}$ is 
lower triangular for the CL scattering kernel. Hence, the
spectral radius of $\hat{\+R}$ in the CL case is
easy to obtain (see Theorem \ref{thm:scll} below).

	\section{The Cercignani-Lampis Scattering Kernel}
	\label{subsec:32}
	For the CL scattering kernel, we will give the recursion formula
to calculate the coefficient matrices in \eqref{eq:gBC5}. 
The CL scattering kernel $K(\+\xi'\rightarrow\+\xi)$ in \eqref{eq:defR}
has the separability and can write as
\[
	K(\+\xi'\rightarrow\+\xi) = K_1(\xi_1,\xi_1')K_2(\xi_2,\xi_2')
	K_3(\xi_3,\xi_3'),
\]
where for $i=1$ and $i=3$,
\begin{eqnarray}
	K_i(\xi_i,\xi_i') = \frac{1}{\sqrt{2\pi\alpha_t(2-\alpha_t)}}
	\exp\left(-\frac{|\xi_i-(1-\alpha_t)\xi_i'|^2}
			{2\alpha_t(2-\alpha_t)}\right)\omega_0(\xi_i'),
\end{eqnarray}
and 
\begin{eqnarray}
	K_2(\xi_2,\xi_2') = \frac{1}{\sqrt{2\pi}\alpha_n}
	I_0\left(\frac{\sqrt{1-\alpha_n}}{\alpha_n}\xi_2\xi_2'\right)
	\exp\left(-\frac{\xi_2^2+\xi_2'^2}{2\alpha_n}\right).
\end{eqnarray}
The Hermite polynomials $\phi_{\+\alpha}$ are isotropic and
can write as
\[\phi_{\+\alpha}(\+\xi)=\phi_{\alpha_1}(\xi_1)
\phi_{\alpha_2}(\xi_2)
\phi_{\alpha_3}(\xi_3)
,\ \omega(\+\xi)=\prod_{i=1}^3\omega_0(\xi_i),\]
where we denote by $\phi_{\alpha_i}=\phi_{\alpha_i}(\xi_i)
=\phi_{\alpha_i\+e_i}(\+\xi)$ and 
$\omega_0(\xi_i)=(\sqrt{2\pi})^{-1}\exp(-\xi_i^2/2)$.
The task is to calculate the integrals $S(\+\alpha,\+\beta)$
and $R(\+\alpha,\+\beta)$ in \eqref{eq:defS} and \eqref{eq:defR}.
Due to the separability, these integrals are products of
two-dimensional integrals. Here $R(\+\alpha,\+\beta)$
is a combination of the modified Bessel function, Hermite 
polynomials, exponentials, and powers.

According to the orthogonality of Hermite polynomials, we have
\begin{equation}
	S(\+\alpha,\+\beta)
	= \delta_{\alpha_1,\beta_1}\delta_{\alpha_3,
	\beta_3}S_0(\alpha_2,\beta_2), 
\end{equation}
where for $\alpha_2,\beta_2\in\bbN,$ we define the half-space 
integral
\begin{equation}
	\label{eq:def_S0}
S_0(\alpha_2,\beta_2) = \int_0^{+\infty}\!\!
	\xi_2\phi_{\alpha_2}\phi_{\beta_2}
	\omega_0(\xi_2)\, \mathrm{d}\xi_2.\
\end{equation}
A direct calculation gives (see Appendix \ref{app:B})
\begin{lemma}
	When $\alpha_2$ is even and $\beta_2$ is odd, we have
	\[
		S_0(\alpha_2,\beta_2) = \frac{1}{2}\left(
		\sqrt{\alpha_2}\delta_{\beta_2,\alpha_2-1}+
		\sqrt{\alpha_2+1}\delta_{\beta_2,\alpha_2+1}\right).
	\]
	When $\alpha_2$ and $\beta_2$ are both even, we have
	\[
		S_0(\alpha_2,\beta_2) = \frac{1}{\sqrt{2\pi}}
		\frac{\alpha_2+\beta_2+1}{1-(\alpha_2-\beta_2)^2}
		z_{\alpha_2}z_{\beta_2},
	\]
	where $z_0=1$ and $z_{n+1}=-\sqrt{n}z_{n-1}/\sqrt{n+1}.$
\end{lemma}

For the CL BCs, the integral $R(\+\alpha,\+\beta)$ involves
the half-space integrals about $K_i(\xi_i,\xi_i')$. When
$i=1$ and $i=3$, the corresponding integrals are explicitly
given in virtue of the symmetry of $K_i(\xi_i,\xi_i')$
(see Appendix \ref{app:C}).

\begin{lemma}
	For $\alpha_1,\beta_1\in\bbN$, the integral
\begin{eqnarray}
	\label{eq:cllT}
T(\alpha_1,\beta_1) = 
	\int_{\bbR^2}\!\!
	\phi_{\beta_1}(\xi_1')
	\phi_{\alpha_1}(\xi_1)K_1(\xi_1,\xi_1')
	\,\mathrm{d}\xi_1'\mathrm{d}\xi_1
\end{eqnarray}
has the explicit expression
	\[
		T(\alpha_1,\beta_1) =
		\delta_{\alpha_1,\beta_1}
		(1-\alpha_t)^{\alpha_1}.
	\]
\end{lemma}
On the other hand, we do not find a ready-made formula 
to calculate the integral about $K_2(\xi_2,\xi_2')$.
For even $\alpha_2$, we denote by 
\begin{eqnarray}
	N(\alpha_2;\xi_2') = \exp\left(\frac{\xi_2'^2}{2\alpha_n}\right)
	\int_{\xi_2>0}\!\!\xi_2
	K_2(\xi_2,\xi_2')\phi_{\alpha_2}(\xi_2)
	\,\md\xi_2
\end{eqnarray}
and claim that (proof in Appendix \ref{app:S})
\begin{lemma}
	\label{lem:03}
	$N(\alpha_2;\xi_2')$ can be represented as
	\[
		N(\alpha_2;\xi_2') =
		\omega_0(\xi_2')
		\exp\left(\frac{\xi_2'^2}{2\alpha_n}\right)
		\sum_{\beta_2=0}^{\alpha_2}
		r_{\alpha_2,\beta_2}\phi_{\beta_2}(\xi_2'),
	\]
	where $r_{\alpha_2,\beta_2}$ are constant coefficients
	which are zero when $\beta_2$ is odd. The recursion
	relation of $r_{\alpha_2,\beta_2}$ is given in Appendix
	\ref{app:S}.
	In particular, we have
	\[ r_{\alpha_2,\alpha_2} = (1-\alpha_n)^{\alpha_2}.\]
\end{lemma}
Due to the orthogonality of Hermite polynomials,
we immediately have 
\begin{eqnarray}
	\label{eq:R}
	R(\+\alpha,\+\beta) = 	\delta_{\alpha_1,\beta_1}
		\delta_{\alpha_3,\beta_3}
		(1-\alpha_t)^{\alpha_1+\alpha_3}r_{\alpha_2,\beta_2},
\end{eqnarray}
where $r_{\alpha_2,\beta_2}=0$ for $\beta_2>\alpha_2.$
So we can conclude that 

\begin{theorem}
	\label{thm:scll}
Suppose $M\geq 3$. For the CL scattering kernel, the 
matrix $\hat{\+R}$ in \eqref{eq:gBC5} is lower triangular.
The spectral radius is 
	\[\rho(\hat{\+R})= \max(|1-\alpha_t|,1-\alpha_n).\]
Hence, when $0<\alpha_t<2$ and $0<\alpha_n\leq 1$,
the conditions in Lemma \ref{lem:02} hold for the CL BCs.
\end{theorem}
\begin{proof}
	Recalling that $\hat{\+R}$ is a rank-one modification
	of $\+R_{M-1,M-1}$ (see \eqref{eq:gBC5}) and the elements
	of $\+R_{M-1,M-1}$ are $R(\+\alpha,\+\beta)$ with
	$\alpha_2$ and $\beta_2$ even.
	When $\alpha_2$ and $\beta_2$ are both even, 
	by definition in Appendix \ref{app:A}, $\+\beta$ is
	ordered after $\+\alpha$ if and
	only if $\alpha_i\neq\beta_i$ for any $i=1,3$,
	or $\alpha_i=\beta_i$ for $i=1,3,$ but $\beta_2>\alpha_2$.
	In both cases, from \eqref{eq:R}, we have
	$R(\+\alpha,\+\beta)=0$. Hence, $\+R_{M-1,M-1}$ is
	lower triangular and so is $\hat{\+R}.$

	The eigenvalues of the lower triangular matrix
	are their diagonal elements. By Lemma \ref{lem:03} and \eqref{eq:R},
	we have
	\[
		R(\+\alpha,\+\alpha) = (1-\alpha_t)^{\alpha_1+\alpha_3}
								(1-\alpha_n)^{\alpha_2}.
	\]
	So $\+R_{M-1,M-1}$ must have the eigenvalue one, corresponding
	to $R(\+0,\+0).$ 
	But the rank-one modification removes the eigenvalue one, which
	gives
	\[
		\hat{\+R}[{\+\alpha},{\+\alpha}] 
		= (1-\alpha_t)^{\alpha_1+\alpha_3}
			(1-\alpha_n)^{\alpha_2}-\delta_{\+\alpha,\+0}.
	\]
	So when $0<\alpha_t<2$ and $0<\alpha_n\leq 1$, we have
	$\rho(\hat{\+R})<1.$ By definition, the condition 
	\eqref{eq:DB} in Lemma \ref{lem:02} also holds.
\end{proof}

Combined with the above theorem, Lemma \ref{lem:02} and
Theorem \ref{thm:31}, we obtain the solvability of the
CL BCs \eqref{eq:gBC5}. 

\begin{exmp}[Fully diffuse BCs]
	Now $\alpha_t=\alpha_n=1$ and the conditions in 
	Lemma \ref{lem:02} hold.
	From \eqref{eq:R} and the recursion relation of
	$r_{\alpha_2,\beta_2}$ in Appendix \ref{app:S}, we
	find that
	$\hat{\+R}=\+0.$
\end{exmp}

\begin{exmp}[Specular BCs]
This is a limiting case when $\alpha_t\rightarrow 0$ and
$\alpha_n\rightarrow 0.$ Now according to Theorem \ref{thm:scll},
the conditions in Lemma \ref{lem:02} do not hold. 
However, the direct calculation from \eqref{eq:R}
shows that $\+R_{M-1,M-1}$ should be an identity matrix.
So the original BCs \eqref{eq:gBC4} give
	\[
	\+M_{M-1,M}(\+w_o+\bar{\+w}_o)=\+0 \quad\Rightarrow
	\quad \+w_o+\bar{\+w}_o=\+0.
	\]
We can easily check the solvability in this special case.
\end{exmp}

In conclusion, for the CL scattering kernel, all the integrals
in \eqref{eq:gBC5} are determined by explicit recursion relations.
The procedure avoids numerical integration and the evaluation
of the modified Bessel function. The total computation cost of
$S(\+\alpha,\+\beta)$ and $R(\+\alpha,\+\beta)$ is $O(M^2)$ operations. 
For more general BCs, we may have to utilize numerical integration
to calculate the half-space integrals,
where Gauss-Hermite quadratures or other ununiform quadratures 
should be cleverly used.

	\section{Slip and Jump Coefficients}
	\label{sec:4}
	\def\dpath{data}
	\newcommand\dvel{(\partial U/\partial x_2)_{\infty}}
\newcommand\dtem{(\partial T/\partial x_1)_{\infty}}
\newcommand\dtems{(\partial T/\partial x_2)_{\infty}}

\subsection{Mathematical formulation}
In the spirit of the Chapman-Enskog expansion, \cite{CC1989} the 
half-space problems have the same governing equations
\eqref{eq:KL} but different driven terms $\bar{\+w}$ in 
the BCs \eqref{eq:gBC5}.

In Kramers' problem, \cite{Kramers1949} the tangential
flow in the Knudsen layer is assumed to be driven by the 
normal velocity gradient at infinity. So Kramers'
problem is defined as the layer equations \eqref{eq:KL}
with the BCs, denoting by 
$\+B=[\+M_{M,M}^T,\+H]\in\bbR^{n\times N}$ in \eqref{eq:gBC5},
\begin{eqnarray}
	\label{eq:Kra}
	\+B\left(\+w-\dvel\+z_k+\bar{\+w}_0\right) = \+0. 
\end{eqnarray}
Here $\dvel$ is a given constant and
$\bar{\+w}_0\in\mathrm{Null}(\+Q)$ with 
$\+\varphi_2^T\bar{\+w}_0=\bar{\+w}_0[{\+e_2}]=0.$
The given driven term $\+z_k\in\bbR^N$ belongs to 
$\mathrm{Null}(\+Q)^{\perp}$ and is derived from the
Chapman-Enskog expansion, satisfying
\[\+Q\+z_k = \+r_{12},\]
where $\+Q$ is the coefficient matrix in \eqref{eq:KL},
corresponding to the linearized Boltzmann operator, and 
$\+r_{12}\in\bbR^N$ is induced from $\mathbb{I}_M$ with
$\+r_{12}[{\+e_1+\+e_2}]=1.$
Due to Theorem \ref{thm:31}, we can solve a linear relation
between $\+\varphi_1^T\bar{\+w}_0$ and $\dvel$.
Denote by $\bar{u}=\+\varphi_1^T\bar{\+w}_0$, then we 
define the viscous-slip coefficient $\zeta_0$ as
\begin{equation}
	\label{eq:etavs}
	\zeta_0 = \frac{\bar{u}}{\sqrt{2}\gamma_1\dvel},
	\quad \gamma_1=\+r_{12}^T\+z_k.
\end{equation}
Here $\gamma_1$ is a normalized viscosity coefficient
and $\zeta_0$ is the ratio of the slip length
to the mean free path. In the Knudsen layer,
the normalized velocity profile $u_d(y)$ is defined by
\[
u_d(y) = -\frac{{u}_{1}(y)}{\sqrt{2}\gamma_1\dvel},
\]
where $u_1(y)=\+w[{\+e_1}]$.

Analogously, the thermal creep problem considers the
tangential flow driven by the tangential gradient of the
temperature. According to the Chapman-Enskog procedure,
we define the thermal creep problem as
the layer equations \eqref{eq:KL} with the BCs
\begin{eqnarray}
	\label{eq:Cre}
	\+B\left(\+w-\dtem\+z_c+\bar{\+w}_0\right) = \+0, 
\end{eqnarray}
where $\dtem$ is a given constant and 
$\+z_c\in\mathrm{Null}(\+Q)^{\perp}$ is given as the solution of
\[\+Q\+z = \+s_{1}.\]
Here $\+s_1$ is induced from $\mathbb{I}_M$ with
$\+s_1[{3\+e_1}]=\sqrt{3/2},\
\+s_1[{\+e_1+2\+e_d}]=\sqrt{1/2},\ d\neq 1.$
Similarly we define the thermal-slip coefficient as
\begin{equation}
	\label{eq:etat}
	\zeta_1= \frac{\bar{u}}{2\gamma_2\dtem},\quad
	\gamma_2 = \frac{2}{5}\+s_1^T\+z_c,
\end{equation}
and the normalized velocity in the Knudsen layer is
\[
	u_d(y) = -\frac{{u}_{1}(y)}{2\gamma_2\dtem}.
\]
Here $\gamma_2$ can be seen as the normalized 
thermal conductivity coefficient.

The temperature jump problem \cite{Welander1954} is a thermal version
of Kramers' problem. The problem considers the heat transfer
in the Knudsen layer driven by the normal temperature
gradient of the bulk flow. This gives the layer equations
\eqref{eq:KL} with the BCs
\begin{eqnarray}
	\label{eq:Temp}
	\+B\left(\+w-\dtems\+z_t+\bar{\+w}_0\right) = \+0, 
\end{eqnarray}
where $\dtems$ is a given constant and $\+z_t
\in\mathrm{Null}(\+Q)^{\perp}$ is given as the solution of
\[\+Q\+z = \+s_{2}.\]
The non-zero entries of $\+s_2$ are 
$\+s_2[{3\+e_2}]=\sqrt{3/2},\
\+s_2[{\+e_2+2\+e_d}]=\sqrt{1/2},\ d\neq 2.$ 
Now we let $\bar{\+w}[{2\+e_i}] 
= \bar{\theta}/{\sqrt{2}}$ and the temperature-jump
coefficient is analogously defined as
\begin{equation}
	\label{eq:zeta3}
	\zeta_2 = \frac{\bar{\theta}}{\sqrt{2}\gamma_2\dtems}.
\end{equation}
The relation \eqref{eq:macro_v} gives
$\displaystyle\sum_{i=1}^3\+w[{2\+e_i}] = \displaystyle
\frac{3\sqrt{2}}{2}{\theta}(y)$ and we define the normalized
temperature in the Knudsen layer as
\begin{equation*}
	\theta_d(y) = -\frac{\theta(y)}{\sqrt{2}\gamma_2\dtems}.
\end{equation*}

The above half-space problems can be solved universally
by the method proposed in our early work. \cite{Yang2022a}
The method is analytical in principle and concluded as follows:
\begin{enumerate}
	\item Solve the generalized eigenvalue problem of
		$(\+A_2,\+Q)$ to obtain general solutions to
		\eqref{eq:KL}. 
	\item Plug the necessary conditions such that 
		the general solution of \eqref{eq:KL} satisfies
		$\+w(\infty)=\+0$ into the BCs \eqref{eq:gBC5}. 
		Solve the obtained linear system to get $\+w(0)$.
	\item Obtain $\+w(y)$ according to the expressions
		of the general solution to \eqref{eq:KL}.
\end{enumerate}

The calculation of the matrices
$\+Q$ and $\+B$ would be another challenge because, by definition,
their elements are high-dimensional integrals (eight-dimensional
for $\+Q$ and six-dimensional for $\+B$). We use the
ready-made $\+Q$ calculated by numerical integration 
in Ref.\onlinecite{Wang2019} for IPL potentials. The matrix $\+B$
for the CL scattering kernel is calculated by the recursion
relations in Sec.\ref{subsec:32}. As will be shown below,
the analytical procedure
provides us with not only an accurate result of $\+B$, but also
explicit expressions of the slip and jump coefficients in terms of
the ACs.

To reduce the computation cost, 
we consider the BGK-type approximation of $\+Q$. \cite{Wang2019}
Namely, for a given constant $L$,
$\+Q[{\+\alpha},{\+\beta}]$ are exactly calculated
from the linearized Boltzmann operator when $|\+\alpha|,|\+\beta|\leq
L$, while the remaining part of $\+Q$ is approximated
by a diagonal matrix.
Combined with the sparsity pattern of $\+A_2,\ \+B$ and $\bar{\+w}$,
we can extract $O(ML^2)$ effective equations from the total $O(M^3)$
equations in \eqref{eq:KL}. The reduced procedure
has been described in detail in Ref.\onlinecite{Yang2022a}, 
which makes the layer
equations \eqref{eq:KL} computable when $M$ is large.
In this paper, we let $L=20$ and $M$ may be several hundred.

	\subsection{Numerical results}
As far as we know, very little data on slip and jump
coefficients are available for the IPL intermolecular 
potentials and CL gas-surface interaction. The subsection
validates the efficiency and accuracy of the moment model
in the above cases for different moment order $M$.

We first consider the hard-sphere (HS) case where $\eta=\infty$ in
the IPL model. In Fig.\ref{fig:01}, we let $\alpha_t=\alpha_n=1$ and 
calculate the coefficients $\zeta_i,i=0,1,2,$ for $M$ ranging from
$5$ to $84$. The label ``$N$'' represents the results given by the
modified BCs \eqref{eq:gBC5}, while ``$G$'' represents the results
given by the Grad BCs \eqref{eq:gBC4}. We see that all the 
coefficients will converge when $M$ becomes larger. The 
results approach their limit from both ends according to
the parity of $M$ for the modified BCs, while only from 
below in the case of the Grad BCs.

Note that the instability of
\eqref{eq:gBC4}, if exists, is about the non-homogeneous term in the layer
equations, \cite{Yang2022b} which would not affect the solvability
of homogeneous equations. In comparison, we
can see that the boundary stabilization does not affect the accuracy
too much. 

\begin{figure}[!htb]
\pgfplotsset{width=0.45\textwidth}
\centering
\begin{tikzpicture} 
\begin{axis}[
    xlabel=$M$, 
	ylabel=$\zeta_{0}$, 
    tick align=outside, 
    legend style={at={(0.73,0.3)},anchor=north} 
    ]
\addplot[smooth,mark=*,blue] table {\dpath/N_t1n1_zeta0.dat};
\addlegendentry{N}
\addplot[smooth,mark=square,red] table {\dpath/G_t1n1_zeta0.dat};
\addlegendentry{G}
\end{axis}
\end{tikzpicture}
\hskip 3pt
\begin{tikzpicture} 
\begin{axis}[
    xlabel=$M$, 
	ylabel=$\zeta_{1}$, 
    tick align=outside, 
    legend style={at={(0.73,0.3)},anchor=north} 
    ]
\addplot[smooth,mark=*,blue] table {\dpath/N_t1n1_zeta1.dat};
\addlegendentry{N}
\addplot[smooth,mark=square,red] table {\dpath/G_t1n1_zeta1.dat};
\addlegendentry{G}
\end{axis}
\end{tikzpicture}
\hskip 3pt
\begin{tikzpicture} 
\begin{axis}[
    xlabel=$M$, 
	ylabel=$\zeta_{2}$, 
    tick align=outside, 
    legend style={at={(0.73,0.3)},anchor=north} 
    ]
\addplot[smooth,mark=*,blue] table {\dpath/N_t1n1_zeta2.dat};
\addlegendentry{N}
\addplot[smooth,mark=square,red] table {\dpath/G_t1n1_zeta2.dat};
\addlegendentry{G}
\end{axis}
\end{tikzpicture}

\caption{Convergence tests of the coefficients $\zeta_i$ in the HS
model, with $\alpha_t=\alpha_n=1$.
}
\label{fig:01}
\end{figure}
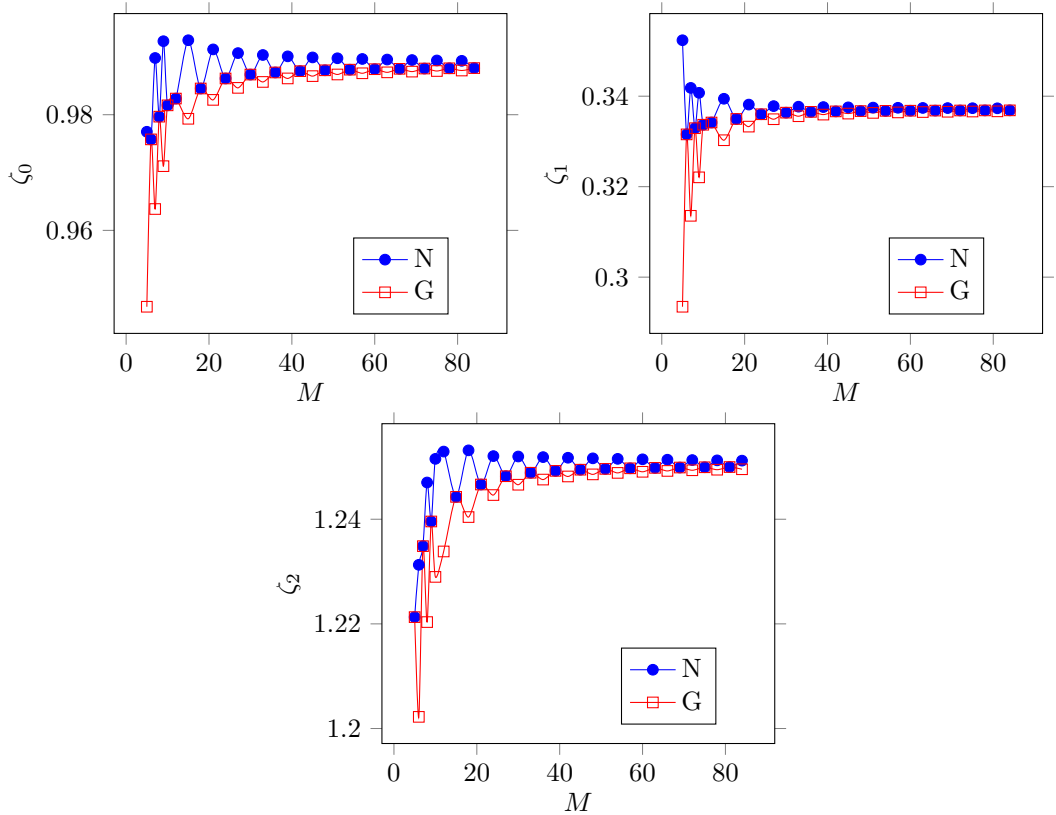

In Fig.\ref{fig:02}, we choose $\alpha_t$ and $\alpha_n$
as their corner values, i.e., $\alpha_t$ close to $0$ or $2$
and $\alpha_n$ close to $0$. The convergence of the viscous-slip
coefficient $\zeta_0$ is also observed.
This shows the stability of our algorithm at the limiting cases.
Similar results are observed for the thermal-slip coefficient
$\zeta_1$ and temperature-jump coefficient $\zeta_2$ (not exhibited
here).

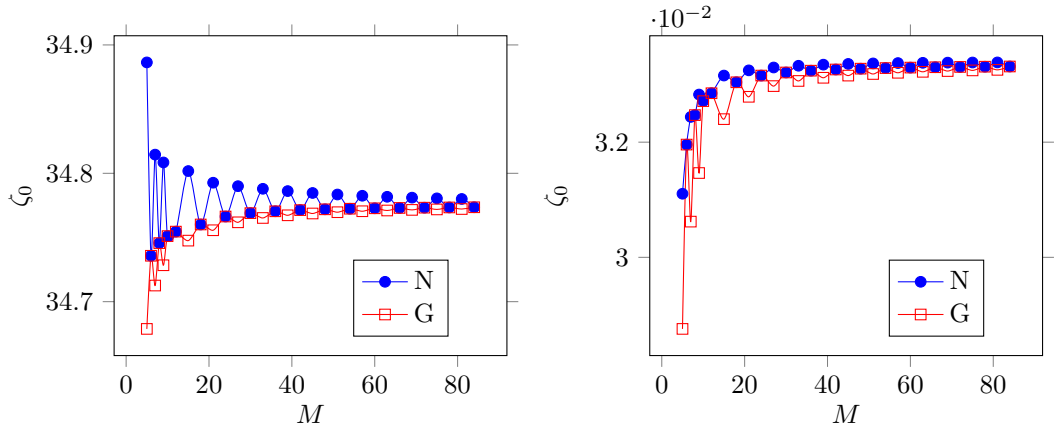
\begin{figure}[!htb]
\pgfplotsset{width=0.45\textwidth}
\centering
\begin{tikzpicture} 
\begin{axis}[
    xlabel=$M$, 
	ylabel=$\zeta_{0}$, 
    tick align=outside, 
    legend style={at={(0.73,0.3)},anchor=north} 
    ]
\addplot[smooth,mark=*,blue] table {\dpath/N_t0n0_zeta0.dat};
\addlegendentry{N}
\addplot[smooth,mark=square,red] table {\dpath/G_t0n0_zeta0.dat};
\addlegendentry{G}
\end{axis}
\end{tikzpicture}
\hskip 3pt
\begin{tikzpicture} 
\begin{axis}[
    xlabel=$M$, 
	ylabel=$\zeta_{0}$, 
    tick align=outside, 
    legend style={at={(0.73,0.3)},anchor=north} 
    ]
\addplot[smooth,mark=*,blue] table {\dpath/N_t2n0_zeta0.dat};
\addlegendentry{N}
\addplot[smooth,mark=square,red] table {\dpath/G_t2n0_zeta0.dat};
\addlegendentry{G}
\end{axis}
\end{tikzpicture}

\caption{Convergence tests of the coefficient $\zeta_0$ in the HS
model. Left: $\alpha_t=0.05,\alpha_n=0.05.$ Right: $\alpha_t=1.95,
\alpha_n=0.05$.
}
\label{fig:02}
\end{figure}

We then study the influence of the moment order. For
$M=4,10,50$, we calculate the slip coefficients $\zeta_i,i=0,1,$ 
for $\alpha_n\in[0,1],\ \alpha_t\in[0.25,2]$. For $M=5,11,51$,
we calculate the temperature-jump coefficient $\zeta_2$ for
$\alpha_n\in[0,1],\ \alpha_t\in[0.25,1].$ Note that when $\alpha_n$
is fixed, $\alpha_t$ and $2-\alpha_t$ would give the same
jump coefficient $\zeta_2$. So we cut the range of $\alpha_t$
in half in the temperature jump problem. 
The results are individually shown in Tab.\ref{tab:01},
Tab.\ref{tab:02}, and Tab.\ref{tab:03}. Because the parity
of $M$ is the same in one table, we expect to see
the convergence from one side in the same table.

\begin{table}[!htb] 
\centering 
\caption{Viscous-slip coefficient $\zeta_0$ for the Cercignani-Lampis
BCs in the HS model}
\label{tab:01}
\resizebox{.85\textwidth}{!}{
\begin{tabular}{ccccccc} 
	$\alpha_t$ & $M$ & $\alpha_n=0$ & $\alpha_n=0.25$ & $\alpha_n=0.5$ &
	$\alpha_n=0.75$ & $\alpha_n=1$ \\ 
\hline
0.25 & 4 & 6.33938 & 6.32513 & 6.31093 & 6.29677 & 6.28265 \\ 
& 10 & 6.37510 & 6.35349 & 6.33403 & 6.31645 & 6.30045 \\ 
& 50 & 6.39435 & 6.36519 & 6.34284 & 6.32364 & 6.30659 \\ 
& Ref.\onlinecite{Su2019} & & 6.365427 & 6.343336 & 6.324267 & 6.307321 \\
0.5 & 4 & 2.77490 & 2.76541 & 2.75593 & 2.74648 & 2.73704 \\ 
& 10 & 2.80361 & 2.78951 & 2.77674 & 2.76513 & 2.75449 \\ 
& 50 & 2.81737 & 2.79906 & 2.78456 & 2.77193 & 2.76060 \\ 
& Ref.\onlinecite{Su2019} & & 2.799516 & 2.785158 & 2.772602 & 2.761338 \\
0.75 & 4 & 1.57406 & 1.56933 & 1.56459 & 1.55987 & 1.55515 \\ 
& 10 & 1.59673 & 1.58982 & 1.58354 & 1.57779 & 1.57249 \\ 
& 50 & 1.60617 & 1.59751 & 1.59046 & 1.58423 & 1.57859 \\ 
& Ref.\onlinecite{Su2019} & & 1.598122 & 1.591127 & 1.584932 & 1.579323 \\
1.0 & 4 & 0.964293 & 0.964293 & 0.964293 & 0.964293 & 0.964293 \\ 
& 10 & 0.981622 & 0.981622 & 0.981622 & 0.981622 & 0.981622 \\ 
& 50 & 0.987722 & 0.987722 & 0.987722 & 0.987722 & 0.987722 \\ 
& Ref.\onlinecite{Su2019} & & 0.988451 & 0.988451 & 0.988451 & 0.988451 \\
1.25 & 4 & 0.590920 & 0.595624 & 0.600333 & 0.605047 & 0.609765 \\ 
& 10 & 0.603435 & 0.610081 & 0.616170 & 0.621810 & 0.627080 \\ 
& 50 & 0.607019 & 0.614835 & 0.621529 & 0.627585 & 0.633179 \\ 
& Ref.\onlinecite{Su2019} & & 0.615670 & 0.622315 & 0.628343 & 0.633906 \\
1.5 & 4 & 0.335520 & 0.344908 & 0.354315 & 0.363740 & 0.373184 \\ 
& 10 & 0.343598 & 0.356670 & 0.368685 & 0.379875 & 0.390396 \\ 
& 50 & 0.345383 & 0.360302 & 0.373369 & 0.385333 & 0.396489 \\ 
& Ref.\onlinecite{Su2019} & & 0.361248 & 0.374217 & 0.386121 & 0.397213 \\
1.75 & 4 & 0.147115 & 0.161201 & 0.175327 & 0.189493 & 0.203700 \\ 
& 10 & 0.151032 & 0.170357 & 0.188173 & 0.204852 & 0.220632 \\ 
& 50 & 0.151646 & 0.173077 & 0.192246 & 0.210003 & 0.226713 \\ 
& Ref.\onlinecite{Su2019} & & 0.174178 & 0.193187 & 0.210840 & 0.227456 \\
2.0 & 4 & 0 & 0.0188546 & 0.0377750 & 0.0567615 & 0.0758145 \\ 
& 10 & 0 & 0.0254564 & 0.0489856 & 0.0711298 & 0.0922059 \\ 
& 50 & 0 & 0.0274630 & 0.0525196 & 0.0759943 & 0.0982846 \\ 
& Ref.\onlinecite{Su2019} & & 0.028851 & 0.053665 & 0.076984 & 0.099153 \\
\end{tabular}
}

\end{table}

In Tab.\ref{tab:01}, the highly accurate 
viscous-slip coefficients calculated from the 
LBE \cite{Su2019} are included.
We see that when $\alpha_t<1$, the slip
coefficient $\zeta_0$ slightly decreases when $\alpha_n$ increases.
The trend reverses when $\alpha_t>1$. When
$\alpha_t=1$, the slip coefficient $\zeta_0$ does not vary with
$\alpha_n$. When $\alpha_t<1$, the moment model with $M=4$
has already given accurate $\zeta_0$, whose relative errors, compared 
to the results of $M=50$ or Ref.\onlinecite{Su2019}, are less than
$3\%.$ When $\alpha_t$ and
$\alpha_n$ are both large, the lowest order moment model with $M=4$
is not accurate, where the relative errors are larger than $10\%.$ But
the moment model with mild order such as $M=10$ performs well, obtaining
the relative errors less than $5\%$ even in some limiting cases.
The BGK-type approximation of $\+Q$ may be another reason that 
the obtained $\zeta_0$ when $M=50$ slightly
deviates from the results of Ref.\onlinecite{Su2019}, i.e., with the
relative errors less than $0.1\%$ in most cases.

\begin{table}[!htb] 
\centering 
\caption{Thermal-slip coefficient $\zeta_1$ for the Cercignani-Lampis
BCs in the HS model}
\label{tab:02}
\resizebox{.85\textwidth}{!}{
\begin{tabular}{ccccccc} 
	$\alpha_t$ & $M$ & $\alpha_n=0$ & $\alpha_n=0.25$ & $\alpha_n=0.5$ &
	$\alpha_n=0.75$ & $\alpha_n=1$ \\ 
\hline
0.25 & 4 & 0.262909 & 0.281742 & 0.300515 & 0.319229 & 0.337884 \\ 
& 10 & 0.268408 & 0.288734 & 0.308258 & 0.327055 & 0.345200 \\ 
& 50 & 0.270170 & 0.290853 & 0.310756 & 0.329846 & 0.348224 \\ 
& Ref.\onlinecite{Siewert2003c} & 0.26960 & 0.29049 & 0.31041 & 0.32950 &
	0.34787 \\
0.5 & 4 & 0.281290 & 0.293648 & 0.305982 & 0.318290 & 0.330573 \\ 
& 10 & 0.287068 & 0.300106 & 0.312749 & 0.325029 & 0.336978 \\ 
& 50 & 0.289701 & 0.302666 & 0.315457 & 0.327896 & 0.339998 \\ 
	& Ref.\onlinecite{Siewert2003c} & 0.28905 & 0.30221 & 0.31503 &
	0.32748 & 0.33958 \\
0.75 & 4 & 0.303436 & 0.309567 & 0.315691 & 0.321809 & 0.327921 \\ 
& 10 & 0.309606 & 0.315927 & 0.322109 & 0.328162 & 0.334095 \\ 
& 50 & 0.312635 & 0.318802 & 0.325009 & 0.331123 & 0.337129 \\ 
	&Ref.\onlinecite{Siewert2003c} & 0.31206 & 0.31834 & 0.32456 & 
	0.33068 & 0.33668 \\
1.0 & 4 & 0.327544 & 0.327544 & 0.327544 & 0.327544 & 0.327544 \\ 
& 10 & 0.333690 & 0.333690 & 0.333690 & 0.333690 & 0.333690 \\ 
& 50 & 0.336727 & 0.336727 & 0.336727 & 0.336727 & 0.336727 \\ 
	&Ref.\onlinecite{Siewert2003c} & 0.33628 & 0.33628 & 0.33628 
	& 0.33628 &0.33628  \\
1.25 & 4 & 0.351523 & 0.345443 & 0.339357 & 0.333265 & 0.327167 \\ 
& 10 & 0.357081 & 0.351070 & 0.345103 & 0.339177 & 0.333287 \\ 
& 50 & 0.359766 & 0.354113 & 0.348222 & 0.342278 & 0.336326 \\ 
	&Ref.\onlinecite{Siewert2003c}& 0.35944 & 0.35369 & 0.34778 &
	0.34183 & 0.33588 \\
1.5 & 4 & 0.372951 & 0.360884 & 0.348794 & 0.336680 & 0.324542 \\ 
& 10 & 0.377351 & 0.365636 & 0.353923 & 0.342203 & 0.330476 \\ 
& 50 & 0.379360 & 0.368544 & 0.357081 & 0.345375 & 0.333538 \\ 
	&Ref.\onlinecite{Siewert2003c} & 0.37915 & 0.36813 & 0.35663 &
	0.34491 & 0.33306 \\
1.75 & 4 & 0.389169 & 0.371334 & 0.353448 & 0.335511 & 0.317523 \\ 
& 10 & 0.392014 & 0.374964 & 0.357779 & 0.340452 & 0.322989 \\ 
& 50 & 0.393091 & 0.377644 & 0.360983 & 0.343752 & 0.326151 \\ 
	&Ref.\onlinecite{Siewert2003c} & 0.39299 & 0.37723 & 0.36049 &
	0.34323 & 0.32562 \\
2.0 & 4 & 0.397561 & 0.374363 & 0.351085 & 0.327725 & 0.304283 \\ 
& 10 & 0.398938 & 0.377003 & 0.354691 & 0.332010 & 0.308984 \\ 
& 50 & 0.398936 & 0.379474 & 0.358082 & 0.335650 & 0.312491 \\ 
	&Ref.\onlinecite{Siewert2003c}& 0.39894 & 0.37904 & 0.35751 &
	0.33502 & 0.31183\\
\end{tabular}
}

\end{table}

In Tab.\ref{tab:02}, we list thermal-slip coefficients for the CL BCs
and compare our results with the LBE's. \cite{Siewert2003c,Su2020}
We find that the relative errors of the moment model with $M=50$
and Ref.\onlinecite{Siewert2003c} are about $0.1\%$ in 
most cases. Compared to the viscous-slip coefficient $\zeta_0$,
the variation of $\zeta_1$ is smaller when $\alpha_t$ changes.
When $\alpha_t$ is fixed and $\alpha_n$ increases, the thermal-slip
coefficient $\zeta_1$ decreases when $\alpha_t<1$, remaining unchanged
when $\alpha_t=1$, and increases when $\alpha_t>1.$
In the thermal creep problem, the moment model with
$M=4$ is enough to provide
an accurate $\zeta_1$ with the relative errors less than $3\%$
compared with Ref.\onlinecite{Siewert2003c}.

\begin{table}[!htb] 
\centering 
\caption{Temperature-jump coefficient $\zeta_2$ for the Cercignani-Lampis
BCs in the HS model}
\label{tab:03}
\resizebox{.85\textwidth}{!}{
\begin{tabular}{ccccccc} 
	$\alpha_t$ & $M$ & $\alpha_n=0$ & $\alpha_n=0.25$ & $\alpha_n=0.5$ &
	$\alpha_n=0.75$ & $\alpha_n=1$ \\ 
\hline
0.25 & 5 & 9.99768 & 5.67986 & 3.73807 & 2.63537 & 1.92775 \\ 
& 11 & 10.1069 & 5.71826 & 3.76224 & 2.65729 & 1.95137 \\ 
& 51 & 10.1558 & 5.73315 & 3.77200 & 2.66686 & 1.96228 \\ 
& Ref.\onlinecite{Siewert2003c} & 10.151 & 5.7318 & 3.7707 & 2.6655 & 1.9609\\
0.5 & 5 & 5.76656 & 3.80951 & 2.69240 & 1.97267 & 1.47398 \\ 
& 11 & 5.86585 & 3.85425 & 2.71941 & 1.99401 & 1.49444 \\ 
& 51 & 5.90485 & 3.87043 & 2.72920 & 2.00205 & 1.50267 \\ 
& Ref.\onlinecite{Siewert2003c} & 5.9030 & 3.8696 & 2.7282 & 2.0010 & 1.5015\\
0.75 & 5 & 4.57861 & 3.16218 & 2.28852 & 1.69885 & 1.27783 \\ 
& 11 & 4.67225 & 3.20886 & 2.31758 & 1.72129 & 1.29831 \\ 
& 51 & 4.70607 & 3.22518 & 2.32781 & 1.72934 & 1.30606 \\ 
& Ref.\onlinecite{Siewert2003c} & 4.7049 & 3.2245 & 2.3270 & 1.7284 & 1.3050 \\
1.0 & 5 & 4.28115 & 2.98966 & 2.17682 & 1.62125 & 1.22129 \\ 
& 11 & 4.37295 & 3.03675 & 2.20650 & 1.64413 & 1.24193 \\ 
& 51 & 4.40516 & 3.05304 & 2.21688 & 1.65225 & 1.24962 \\ 
& Ref.\onlinecite{Siewert2003c} & 4.4041 & 3.0524 & 2.2161 & 1.6514 & 1.2486 \\
\end{tabular}
}

\end{table}

In Tab.\ref{tab:03}, we compare our temperature-jump coefficients
with the results of the LBE. \cite{Siewert2003c}
We find that the temperature-jump coefficient $\zeta_2$ decreases when
$\alpha_n$ and $\alpha_t\leq 1$ turn larger.
Analogously, the moment model with $M=5$ seems enough
to provide a relatively accurate temperature-jump coefficient.

For the sake of brevity, we do not show the results of other
IPL potentials here, although the comparison has been accomplished
for the Maxwell molecules ($\eta=5$) and some other available
data in Ref.\onlinecite{Su2019,Su2020}. In a word, all the results 
support that a mild moment model with $M\approx 10$ can capture
slip and jump coefficients well even in some limiting cases.

We compare the wall clock time used for moment models
with different $M$. Our code is not optimized and implemented
by MATLAB R2019, running on the laptop with i7-8550U CPU @ 1.80GHz.
In Fig.\ref{fig:04}, the label ``time A'' contains the cost to 
generate $\+B$ in the BCs, solving the
generalized eigenvalue problem of $(\+A_2,\+Q)$ and 
the linear algebraic system, but does not include the generation 
of $\+Q$. That's to say, $\+Q$ is calculated beforehand and 
read from the file.
In reality, when the model is fixed, the eigenvalue problem only needs
to solve once. So we also consider the ``time B'', which does not
contain the cost to solve the generalized eigenvalue problem
compared with the ``time A''.

In Fig.\ref{fig:04}, we repeatly calculate the coefficients
for $100$ times and take the average wall time. We find that
the main cost comes from solving the generalized eigenvalue
problem. The ``time B'' for $M\leq 32$ is less than $1$ms.
When $M=64$ (not shown in the figure), for the viscous-slip
coefficient $\zeta_0$,
the ``time A'' is about $3.57$s and the ``time B'' is about
$0.42$s. Combined with the previous results that a mild order
moment model performs well,
we conclude that the moment model is very efficient and 
accurate in capturing slip and jump coefficients.

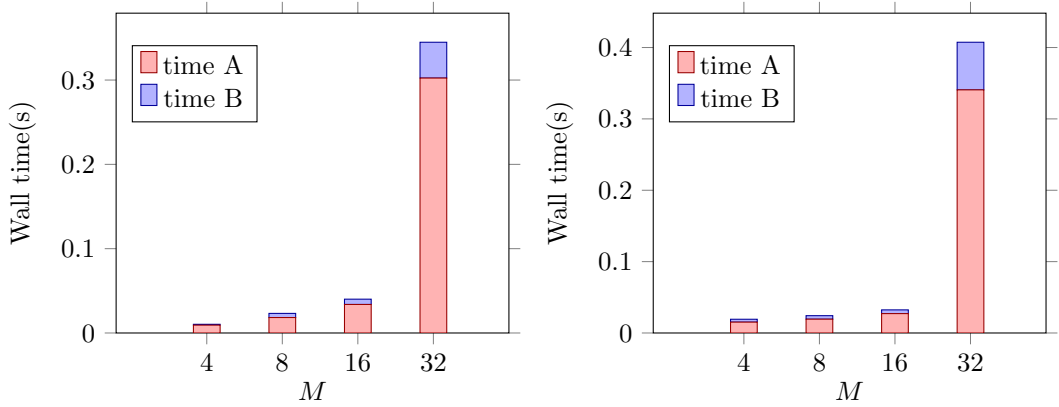
\begin{figure}[!htb]
\pgfplotsset{width=0.45\textwidth}
\centering
\begin{tikzpicture} 
\begin{axis}[
    xlabel=$M$, 
	ylabel=Wall time(s),
    ybar stacked,	
    tick align=outside, 
    legend style={at={(0.2,0.9)},anchor=north},
	xmin = -0.2, xmax = 5, ymin=0,
	xtick = {1,2,3,4},
	xticklabels={4,8,16,32},
    ]
	\addplot+ [ybar,red!30,draw=red!60!black] coordinates {
	(1,0.009331) (2,0.018329) (3,0.033828) (4,0.302353)
	};
	\addplot+ [ybar,blue!30,draw=blue!60!black] coordinates {
	(1,0.000976) (2,0.004910) (3,0.006285) (4,0.042480)
	};
	\legend{time A, time B}
\end{axis}
\end{tikzpicture}
\hskip 3pt
\begin{tikzpicture} 
\begin{axis}[
    xlabel=$M$, 
	ylabel=Wall time(s),
    ybar stacked,	
    tick align=outside, 
    legend style={at={(0.2,0.9)},anchor=north},
	xmin = -0.2, xmax = 5, ymin=0,
	xtick = {1,2,3,4},
	xticklabels={4,8,16,32},
    ]
	\addplot+ [ybar,red!30,draw=red!60!black] coordinates {
	(1,0.015227) (2,0.019362) (3,0.027163) (4,0.340672)
	};
	\addplot+ [ybar,blue!30,draw=blue!60!black] coordinates {
	(1,0.004024) (2,0.004903) (3,0.005222) (4,0.066722)
	};
	\legend{time A, time B}
\end{axis}
\end{tikzpicture}

\caption{The wall time to calculate the slip/jump coefficients
	for different $M$ in the HS model ($\alpha_n=\alpha_t=1$).
	Left: the viscous-slip coefficient $\zeta_0$.
	Right: the temperature-jump coefficient $\zeta_2.$
}
\label{fig:04}
\end{figure}



	\subsection{Explicit expressions}
In this subsetion, we focus on the CL scattering kernel 
and IPL potentials.
The moment model is analytical in the sense that i) the formal
solution of the layer equations \eqref{eq:KL} is available and
ii) the Hermite expansion of CL BCs is explicitly given. 
Therefore, when $M$ is small, we can write explicit expressions
about slip and jump coefficients. These formulae are 
functions of the ACs, i.e., $\alpha_t$ 
and $\alpha_n$ in the CL kernel. 
The coefficients in the formulae are different for different 
IPL intermolecular potentials.

First we consider the viscous-slip coefficient $\zeta_0$.
When $M=2$, the moment model gives the formula
\begin{equation}
	\label{eq:z01}
	\zeta_0 = \frac{2-\alpha_t}{\alpha_t}\frac{\sqrt{\pi}}{2}
\end{equation}
for all the IPL potentials. The formula \eqref{eq:z01} is independent
of the intermolecular potential and the accommodation coefficient
$\alpha_n$, agreeing with the literature.
\cite{Loyalka1967,Zhang2021} Compared with Tab.\ref{tab:01},
\eqref{eq:z01} would have a relatively large
deviation above $10\%$ in the HS model. To alleviate this deviation, 
Ref.\onlinecite{Loyalka1967,Klinc1972,Zhang2021} modify the formula as 
\begin{equation}
	\label{eq:z01-1}
	\zeta_0 = \frac{2-\alpha_t}{\alpha_t}\frac{\sqrt{\pi}}{2}
	\left(1+0.1366\alpha_t\right).
\end{equation}
While in Ref.\onlinecite{Lilley2007,Su2019}, the formula reads as
\begin{equation}
	\label{eq:z00}
	\zeta_0 = \frac{a}{\alpha_t}-b\alpha_t-c,	
\end{equation}
where $a,b$, and $c$ are fitting coefficients relying on
the intermolecular potential and $\alpha_n$.
It's shown \cite{Su2019} that the above fitting formula can
predict the LBE solutions well. 

In comparison, the moment model with $M=4$ gives a new 
explicit expression of $\zeta_0$. We omit the 
tedious process of the calculation here.
The expression is very concise when we 
define some auxiliary parameters as
\begin{equation}
	\label{eq:defm}
m_1=\frac{2-\alpha_t}{\alpha_t},\
	m_2=-1+\frac{2}{\alpha_t(\alpha_t^2-3\alpha_t+3)},\
m_3 = -1+\frac{2}{\alpha_n+\alpha_t-\alpha_n\alpha_t}.
\end{equation}
The viscous-slip coefficient reads as
\begin{equation}
\label{eq:z0m4}
\zeta_0 = \frac{\sqrt{\pi}}{2}m_1 
+ \frac{d_1m_2+d_2m_3+d_3m_2m_3}{c_1m_2+c_2m_3+c_3m_2m_3+c_4},
\end{equation}
where $c_i$ and $d_i$ are constants determined by the 
intermolecular potential. Tab.\ref{tab:04} shows their
values in some special cases. Note that there exists a freedom
about these coefficients, so we may as well assume $c_4=1$.
Compared to the formulae \eqref{eq:z01} and \eqref{eq:z00},
the formula \eqref{eq:z0m4} has stronger 
nonlinearity about $\alpha_n$ and $\alpha_t$. The
accuracy of \eqref{eq:z0m4} in the HS model is shown in Tab.\ref{tab:01},
where the relative errors are less than $3\%$ when $\alpha_t<1.$

\begin{table}[!htb]
\centering 
\caption{Coefficients in \eqref{eq:z0m4} for the IPL models}
\label{tab:04}
\begin{tabular}{cccccccc} 
	$\eta$ & $c_1$ & $c_2$ & $c_3$ & $c_4$ &
	$d_1$ & $d_2$ & $d_3$ \\ 
\hline
$\infty$\ (HS) & 0.9302 & 0.9861 & 0.9135 & 1 &
0.006073 & 0.1506 & 0.1423 \\
$5$\ (Maxwell molecules) & 0.9346 & 0.9868 & 0.9203 & 1 &
0 & 0.2216 & 0.2071 \\
$10$\ (Hard potential) & 0.9318 & 0.9856 & 0.9156 & 1 &
0.001965 & 0.1795 & 0.1673 \\
$3.1$\ (Soft potential) & 0.9363 & 0.9874 & 0.9230 & 1 &
0.005607 & 0.3026 & 0.2921 \\
\end{tabular}

\end{table}

The formula \eqref{eq:z0m4} successfully explains two
phenomena observed in numeric. I). A linear correction 
should be added to improve the accuracy of \eqref{eq:z01}
as in \eqref{eq:z00} and \eqref{eq:z01-1}.
Fig.\ref{fig:100} exhibits the dependence of 
$\zeta_0-\sqrt{\pi}m_1/2$ about $\alpha_t$ when $M=30$ and
$\alpha_n$ is fixed. We can see that the relation is almost linear.
In Ref.\onlinecite{Zhang2021}, the correction \eqref{eq:z01-1} for
the CL case is heuristically given without rigorous derivation.
Here, from Tab.\ref{tab:04}, we may
approximately let $c_i=1$, $d_1=0$ and $d_2=d_3$ 
in \eqref{eq:z0m4}, which gives
\[
\zeta_0	= \frac{\sqrt{\pi}}{2}m_1 + d_2\frac{m_3}{1+m_3}
	= \frac{\sqrt{\pi}}{2}m_1 + \frac{d_2}{2}
	(-(1-\alpha_n)\alpha_t+2-\alpha_n).
\]
So the prediction of \eqref{eq:z0m4} agrees with Fig.\ref{fig:100}
qualitatively.


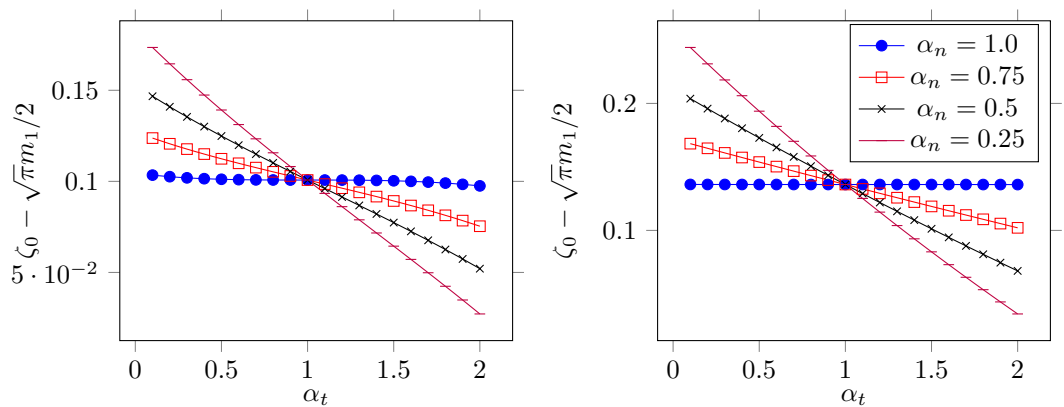
\begin{figure}[!htb]
\pgfplotsset{width=0.45\textwidth}
\centering
\begin{tikzpicture} 
\begin{axis}[
    ylabel=$\zeta_0-\sqrt{\pi}m_1/2$, 
	xlabel=$\alpha_t$, 
    tick align=outside, 
    legend style={at={(0.73,0.99)},anchor=north} 
    ]
\addplot[smooth,mark=*,blue] table {\dpath/zeta0_an1.00.dat};
\addplot[smooth,mark=square,red] table {\dpath/zeta0_an0.75.dat};
\addplot[smooth,mark=x,black] table {\dpath/zeta0_an0.50.dat};
\addplot[smooth,mark=-,purple] table {\dpath/zeta0_an0.25.dat};
\end{axis}
\end{tikzpicture}
\hskip 3pt
\begin{tikzpicture} 
\begin{axis}[
    ylabel=$\zeta_0-\sqrt{\pi}m_1/2$, 
	xlabel=$\alpha_t$, 
    tick align=outside, 
    legend style={at={(0.73,0.99)},anchor=north} 
    ]
\addplot[smooth,mark=*,blue] table {\dpath/5_zeta0_an1.00.dat};
\addlegendentry{$\alpha_n=1.0$}
\addplot[smooth,mark=square,red] table {\dpath/5_zeta0_an0.75.dat};
\addlegendentry{$\alpha_n=0.75$}
\addplot[smooth,mark=x,black] table {\dpath/5_zeta0_an0.50.dat};
\addlegendentry{$\alpha_n=0.5$}
\addplot[smooth,mark=-,purple] table {\dpath/5_zeta0_an0.25.dat};
\addlegendentry{$\alpha_n=0.25$}
\end{axis}
\end{tikzpicture}

\caption{The dependece of $\zeta_0-\sqrt{\pi}m_1/2$ about $\alpha_t$.
	Left: HS model. Right: Maxwell molecules.
}
\label{fig:100}
\end{figure}

The formula \eqref{eq:z0m4} also predicts that II). when
the first term $\sqrt{\pi}m_1/2$ dominates, the 
viscous-slip coefficient is insensitive to
the intermolecular potentials. However, when 
$\alpha_t$ is close to $2$, i.e., in the backward scattering case,
the latter term in \eqref{eq:z0m4} will dominate.
According to the values of $d_2$ and $d_3$ in
Tab.\ref{tab:04}, we know that now the IPL model with
$\eta=3.1$ may have a viscous-slip coefficient
twice as large as the HS model. In Tab.\ref{tab:e01},
we compare the viscous-slip coefficient $\zeta_0$ 
for different IPL potentials when $\alpha_t=0.25,0.5,1.75,2,
\ \alpha_n\in[0.25,1]$ and $M=50.$ We can 
see that the prediction of \eqref{eq:z0m4}
agrees with our numerical results.
The phenomena are also reported in the numerical
results of Ref.\onlinecite{Su2019}. 

\begin{table}[!htb] 
\centering 
\caption{Viscous-slip coefficient $\zeta_0$ for the Cercignani-Lampis
	BCs and IPL models ($M=50$)}
\label{tab:e01}
\begin{tabular}{cccccc} 
	$\alpha_t$ & $\eta$ & $\alpha_n=0.25$ & $\alpha_n=0.5$ &
	$\alpha_n=0.75$ & $\alpha_n=1$ \\ 
	\hline
0.25 & $\infty$ & 6.36519 & 6.34284 & 6.32364 & 6.30659 \\ 
 & 10 & 6.39080 & 6.36387 & 6.34040 & 6.31926 \\ 
 & 5 & 6.42968 & 6.39653 & 6.36717 & 6.34037 \\ 
 & 3.1 & 6.51376 & 6.46837 & 6.42748 & 6.38965 \\ 
0.5 & $\infty$ & 2.79906 & 2.78456 & 2.77193 & 2.76060 \\ 
 & 10 & 2.82086 & 2.80350 & 2.78811 & 2.77408 \\ 
 & 5 & 2.85360 & 2.83237 & 2.81319 & 2.79545 \\ 
 & 3.1 & 2.92349 & 2.89474 & 2.86824 & 2.84328 \\ 
1.75 & $\infty$ & 0.173077 & 0.192246 & 0.210003 & 0.226713 \\ 
 & 10 & 0.177461 & 0.199731 & 0.221007 & 0.241554 \\ 
 & 5 & 0.185033 & 0.211481 & 0.237526 & 0.263340 \\ 
 & 3.1 & 0.203235 & 0.237705 & 0.272854 & 0.308798 \\ 
2 & $\infty$ & 0.0274630 & 0.0525196 & 0.0759943 & 0.0982846 \\ 
 & 10 & 0.0303310 & 0.0592114 & 0.0871847 & 0.114498 \\ 
 & 5 & 0.0343861 & 0.0684421 & 0.102506 & 0.136708 \\ 
 & 3.1 & 0.0426832 & 0.0867627 & 0.132487 & 0.179964 \\ 
\end{tabular}

\end{table}

Then we consider the thermal-slip coefficient $\zeta_1$.
In Fig.\ref{fig:101}, we plot the relation of $\zeta_1$ 
about $\alpha_t$
when $M=30$ and $\alpha_n$ is fixed. 
We can see that $\zeta_1$ is nonlinear about $\alpha_t$ 
and its trends on $\alpha_t$ differ according to the
intermolecular potential. For example when $\alpha_n=1$
and $\alpha_t$ increases, $\zeta_1$ decreases in the HS model
but increases in the soft potential case where $\eta=3.1.$
This phenomenon is also reported in Ref.\onlinecite{Su2020}.

\begin{figure}[!htb]
\pgfplotsset{width=0.45\textwidth}
\centering
\begin{tikzpicture} 
\begin{axis}[
    ylabel=$\zeta_1$, 
	xlabel=$\alpha_t$, 
    tick align=outside, 
    legend style={at={(0.73,0.99)},anchor=north} 
    ]
\addplot[smooth,mark=*,blue] table {\dpath/zeta1_an1.00.dat};
\addplot[smooth,mark=square,red] table {\dpath/zeta1_an0.75.dat};
\addplot[smooth,mark=x,black] table {\dpath/zeta1_an0.50.dat};
\addplot[smooth,mark=-,purple] table {\dpath/zeta1_an0.25.dat};
\end{axis}
\end{tikzpicture}
\hskip 3pt
\begin{tikzpicture} 
\begin{axis}[
    ylabel=$\zeta_1$, 
	xlabel=$\alpha_t$, 
    tick align=outside, 
    legend style={at={(0.73,0.45)},anchor=north} 
    ]
\addplot[smooth,mark=*,blue] table {\dpath/5_zeta1_an1.00.dat};
\addplot[smooth,mark=square,red] table {\dpath/5_zeta1_an0.75.dat};
\addplot[smooth,mark=x,black] table {\dpath/5_zeta1_an0.50.dat};
\addplot[smooth,mark=-,purple] table {\dpath/5_zeta1_an0.25.dat};
\end{axis}
\end{tikzpicture}
\hskip 3pt
\begin{tikzpicture} 
\begin{axis}[
    ylabel=$\zeta_1$, 
	xlabel=$\alpha_t$, 
    tick align=outside, 
    legend style={at={(1.28,0.45)},anchor=north} 
    ]
\addplot[smooth,mark=*,blue] table {\dpath/3_zeta1_an1.00.dat};
\addlegendentry{$\alpha_n=1.0$}
\addplot[smooth,mark=square,red] table {\dpath/3_zeta1_an0.75.dat};
\addlegendentry{$\alpha_n=0.75$}
\addplot[smooth,mark=x,black] table {\dpath/3_zeta1_an0.50.dat};
\addlegendentry{$\alpha_n=0.5$}
\addplot[smooth,mark=-,purple] table {\dpath/3_zeta1_an0.25.dat};
\addlegendentry{$\alpha_n=0.25$}
\end{axis}
\end{tikzpicture}

\caption{The dependece of $\zeta_1$ about $\alpha_t$.
	Left: HS model. Right: Maxwell molecules.
	Bottom: $\eta=3.1.$
}
\label{fig:101}
\end{figure}
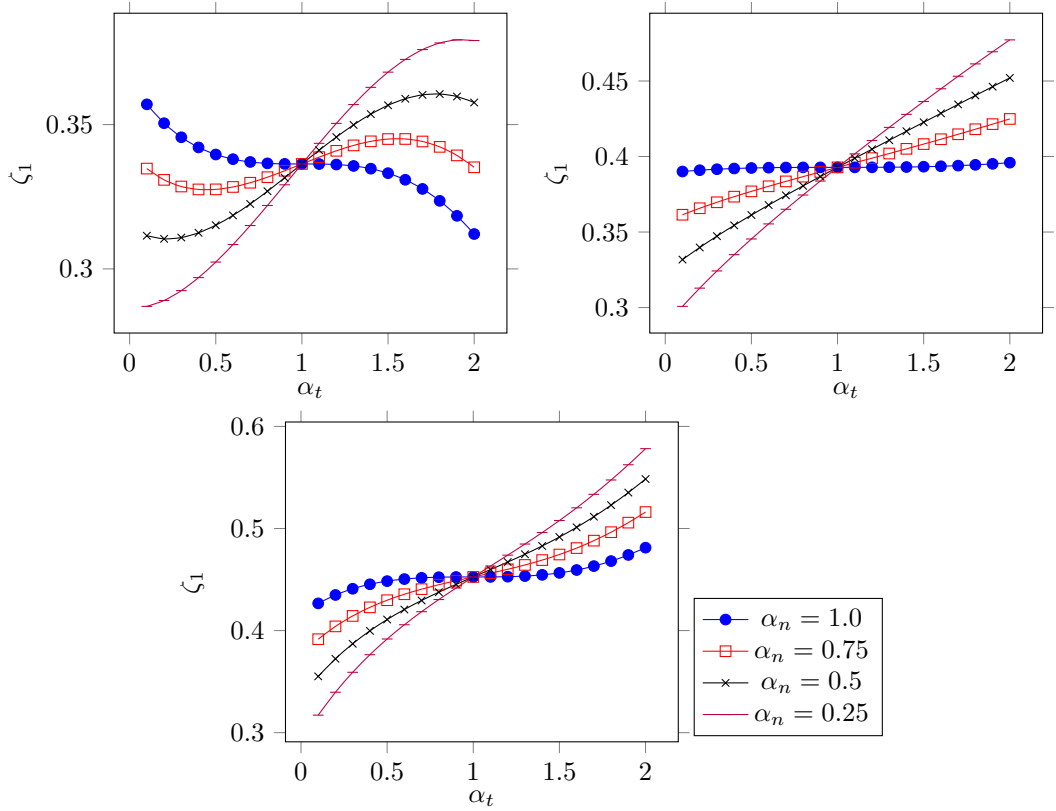

The trends are different from the Maxwell diffuse-specular case,
where $\zeta_1$ is linear about $\chi$ when the other conditions
are fixed. Here $\chi$ is the accommodation coefficient in the 
Maxwell accommodation BCs. For the Maxwell model, if
the scaled thermal-slip coefficient is defined as
\[
	\sigma_T = 	\frac{5}{2}\sqrt{\frac{2}{\pi}}
	\frac{\gamma_2}{\gamma_1}\zeta_1,
\]
then Maxwell gives the first approximation in history as
\[
	\sigma_T = \frac{3}{4},
\]
and Loyalka (cf. Ref.\onlinecite{Loyalka1989} and refs therein)
modified it as
\[
	\sigma_T = \frac{3}{4}\left(1+\frac{1}{2}\chi\right).
\]

For the CL scattering kernel, the explicit formula for
the thermal-slip coefficient is relatively rare. A formula
derived from the variational method for the HS model is given
by Ref.\onlinecite{NN2020}. For the moment model with $M=4$, 
we can derive the analytical expression of $\zeta_1$ as  
\begin{equation}
	\label{eq:z1-1}
	\zeta_1 = \frac{1}{4}+\frac{d_1m_2+d_2m_3+d_3}
		{c_1m_2+c_2m_3+c_3m_2m_3+c_4},
\end{equation}
where $m_2$ and $m_3$ are the same as in \eqref{eq:defm}.
The constant coefficients $c_i,\ d_i$ are determined by
intermolecular potentials, given in Tab.\ref{tab:05} for some
special cases. Here we let $c_4=1$ again.

\begin{table}[!htb]
\centering 
\caption{Coefficients in \eqref{eq:z1-1} for the IPL models}
\label{tab:05}
\begin{tabular}{cccccccc} 
	$\eta$ & $c_1$ & $c_2$ & $c_3$ & $c_4$ &
	$d_1$ & $d_2$ & $d_3$ \\ 
\hline
$\infty$\ (HS) & 0.9302 & 0.9861 & 0.9135 & 1 &
0.1892 & -0.03975 & 0.1476 \\
$5$\ (Maxwell molecules) & 0.9346 & 0.9868 & 0.9203 & 1 &
0.2336 & 0.01564 & 0.2500 \\
$10$\ (Hard potential) & 0.9318 & 0.9856 & 0.9156 & 1 &
0.2085 & -0.01632 & 0.1918 \\
$3.1$\ (Soft potential) & 0.9363 & 0.9874 & 0.9230 & 1 &
0.2751 & 0.07147 & 0.3484 \\
\end{tabular}

\end{table}

Compared with Ref.\onlinecite{NN2020}, the formula \eqref{eq:z1-1} is much
simpler in form and suitable for all IPL potentials. As shown
in Tab.\ref{tab:02}, the formula would
give accurate thermal-slip coefficients with the relative errors
less than $3\%$ in the HS case.
The formula \eqref{eq:z1-1} can also
illustrate the phenomenon mentioned
in the discussion of Fig.\ref{fig:101}.
When $\alpha_n=1$, we have $m_3=1,$ which results in
\[
\zeta_1 = \frac{1}{4}+\frac{d_1m_2+d_2+d_3}
		{(c_1+c_3)m_2+c_2+c_4}.
\]
Apparently, if $d_1>0$ and $c_1+c_3>0$, the above $\zeta_1$
is an increasing function about $m_2$ when
\[
	\frac{d_2+d_3}{d_1} < \frac{c_2+c_4}{c_1+c_3}.
\]
According to Tab.\ref{tab:05}, we can check that it is the case
of the HS model and the hard potential with $\eta=10$. While for the 
Maxwell molecules and the soft potential with $\eta=3.1$, the opposite
inequality holds, and $\zeta_1$ is a decreasing function about
$m_2$. Since $m_2$ is a decreasing function about $\alpha_t$,
when $\alpha_n=1$, we have $\zeta_1$ decreasing about $\alpha_t$
in the HS and $\eta=10$ case,
while increasing for the Maxwell molecules and $\eta=3.1$ case.
Unlike \eqref{eq:z0m4}, there seems no dominant term in
\eqref{eq:z1-1} and the dependence of $\zeta_1$ on different
IPL potentials is relatively complicated.
Thermal-slip coefficients for different IPL potentials with $M=50$
are shown in Tab.\ref{tab:e02}. 

\begin{table}[!htb]
\centering 
\caption{Thermal-slip coefficient $\zeta_1$ for the Cercignani-Lampis
	BCs and IPL models ($M=50$)}
\label{tab:e02}
\begin{tabular}{cccccc} 
	$\alpha_t$ & $\eta$ & $\alpha_n=0.25$ & $\alpha_n=0.5$ &
	$\alpha_n=0.75$ & $\alpha_n=1$ \\ 
\hline
0.25 & $\infty$ & 0.290853 & 0.310756 & 0.329846 & 0.348224 \\ 
 & 10 & 0.302713 & 0.324813 & 0.346144 & 0.366798 \\ 
 & 5 & 0.318814 & 0.343789 & 0.368016 & 0.391574 \\ 
 & 3.1 & 0.349837 & 0.380177 & 0.409672 & 0.438373 \\ 
0.5 & $\infty$ & 0.302666 & 0.315457 & 0.327896 & 0.339998 \\ 
 & 10 & 0.321031 & 0.335170 & 0.349038 & 0.362633 \\ 
 & 5 & 0.345619 & 0.361500 & 0.377187 & 0.392660 \\ 
 & 3.1 & 0.391936 & 0.411003 & 0.429919 & 0.448637 \\ 
1.75 & $\infty$ & 0.377644 & 0.360983 & 0.343752 & 0.326151 \\ 
 & 10 & 0.412165 & 0.394188 & 0.375274 & 0.355639 \\ 
 & 5 & 0.457493 & 0.437730 & 0.416661 & 0.394473 \\ 
 & 3.1 & 0.540527 & 0.517270 & 0.492286 & 0.465637 \\ 
2 & $\infty$ & 0.379474 & 0.358082 & 0.335650 & 0.312491 \\ 
 & 10 & 0.422007 & 0.399075 & 0.374592 & 0.348866 \\ 
 & 5 & 0.477440 & 0.452313 & 0.425143 & 0.396153 \\ 
 & 3.1 & 0.578318 & 0.548573 & 0.516248 & 0.481319 \\ 
\end{tabular}

\end{table}

Finally we focus on the temperature-jump coefficient $\zeta_2.$
For the Maxwell diffuse-specular BCs, $\zeta_2$ has explicit
formulae like \eqref{eq:z00}, i.e., a linear combination of
$(2-\chi)/\chi,\ \chi$ and $1$. Corresponding formulae are given by
Maxwell, \cite{Maxwell} Welander, \cite{Welander1954} and
Loyalka. \cite{Loyalka1978} There is less work about
expressions of $\zeta_2$ with the IPL intermolecular
potentials and CL BCs.

When $M=3$, we derive the analytical expression of $\zeta_2$
from the moment model. In this case, $\zeta_2$ is independent
of the intermolecular potential and reads as
\begin{equation}
	\label{eq:z2}
	\zeta_2 = \frac{\displaystyle\frac{\sqrt{15}}{5}\left(\frac{9}{4}n_1
	+n_2\right)+\frac{5\sqrt{2}}{8}\sqrt{2\pi}n_1n_2}
	{n_1+n_2+\displaystyle\sqrt{\frac{30}{2\pi}}},
\end{equation}
where
\begin{equation}
	\label{eq:defn}
	n_1 = \frac{2-\alpha_n}{\alpha_n},\quad
	n_2 = -1+\frac{2}{\alpha_t(2-\alpha_t)}.
\end{equation}
The formula \eqref{eq:z2} is similar as the expressions 
in Ref.\onlinecite{Struch2013,Zhang2021}. It shows that $\zeta_2$
would increase when $\alpha_n$ and $\alpha_t$ decrease.
This coincides with the numerical results in Fig.\ref{fig:301},
where the dependence of $\zeta_2$ on $\alpha_n$ and $\alpha_t$
is exhibited for the HS model when $M=31$. 

\begin{figure}[!htb]
\pgfplotsset{width=0.45\textwidth}
\centering
\begin{tikzpicture} 
\begin{axis}[
    ylabel=$\zeta_2$, 
	xlabel=$\alpha_t$, 
    tick align=outside, 
    legend style={at={(0.73,0.99)},anchor=north} 
    ]
\addplot[smooth,mark=*,blue] table {\dpath/zeta2_an1.00.dat};
\addlegendentry{$\alpha_n=1.0$}
\addplot[smooth,mark=square,red] table {\dpath/zeta2_an0.75.dat};
\addlegendentry{$\alpha_n=0.75$}
\addplot[smooth,mark=x,black] table {\dpath/zeta2_an0.50.dat};
\addlegendentry{$\alpha_n=0.5$}
\addplot[smooth,mark=-,purple] table {\dpath/zeta2_an0.25.dat};
\addlegendentry{$\alpha_n=0.25$}
\end{axis}
\end{tikzpicture}
\hskip 3pt
\begin{tikzpicture} 
\begin{axis}[
    ylabel=$\zeta_2$, 
	xlabel=$\alpha_n$, 
    tick align=outside, 
    legend style={at={(0.73,0.99)},anchor=north} 
    ]
\addplot[smooth,mark=*,blue] table {\dpath/zeta2_at1.00.dat};
\addlegendentry{$\alpha_t=1.0$}
\addplot[smooth,mark=square,red] table {\dpath/zeta2_at0.75.dat};
\addlegendentry{$\alpha_t=0.75$}
\addplot[smooth,mark=x,black] table {\dpath/zeta2_at0.50.dat};
\addlegendentry{$\alpha_t=0.5$}
\addplot[smooth,mark=-,purple] table {\dpath/zeta2_at0.25.dat};
\addlegendentry{$\alpha_t=0.25$}
\end{axis}
\end{tikzpicture}

\caption{The dependece of $\zeta_2$ in the HS model.
}
\label{fig:301}
\end{figure}
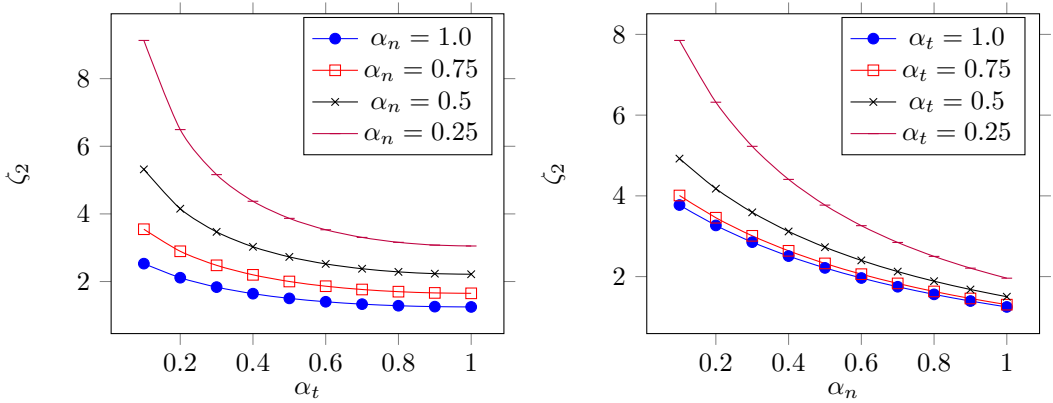

\begin{table}[!htb]
\centering 
\caption{Temperature-jump coefficient $\zeta_2$ 
	for different IPL models ($M=51$)}
\label{tab:07}
\begin{tabular}{ccccccc} 
	  $\alpha_t$ & $\eta$ & $\alpha_n=0$ & $\alpha_n=0.25$ 
	  & $\alpha_n=0.5$ & $\alpha_n=0.75$ & $\alpha_n=1$ \\ 
	 \hline
0.25 & $\infty$ & 10.1558 & 5.73315 & 3.77200 & 2.66686 & 1.96228 \\ 
& 10 & 10.2440 & 5.78167 & 3.80694 & 2.69542 & 1.98856 \\ 
& 5 & 10.3696 & 5.85382 & 3.85960 & 2.73847 & 2.02783 \\ 
& 3.1 &10.6213 & 6.00574 & 3.97211 & 2.83055 & 2.11117 \\ 
0.5 & $\infty$ & 5.90485 & 3.87043 & 2.72920 & 2.00205 & 1.50267 \\ 
& 10 & 5.98305 & 3.91934 & 2.76441 & 2.03034 & 1.52837 \\ 
& 5 & 6.09432 & 3.99076 & 2.81680 & 2.07283 & 1.56697 \\ 
& 3.1 &6.31497 & 4.13759 & 2.92688 & 2.16318 & 1.64920 \\ 
0.75 & $\infty$ & 4.70607 & 3.22518 & 2.32781 & 1.72934 & 1.30606 \\ 
& 10 & 4.77847 & 3.27307 & 2.36293 & 1.75777 & 1.33201 \\ 
& 5 & 4.88165 & 3.34261 & 2.41483 & 1.80028 & 1.37094 \\ 
& 3.1 &5.08564 & 3.48431 & 2.52300 & 1.89021 & 1.45378 \\ 
1.0 & $\infty$ & 4.40516 & 3.05304 & 2.21688 & 1.65225 & 1.24962 \\ 
& 10 & 4.47565 & 3.10043 & 2.25189 & 1.68071 & 1.27568 \\ 
& 5 & 4.57622 & 3.16918 & 2.30354 & 1.72322 & 1.31477 \\ 
& 3.1 &4.77494 & 3.30897 & 2.41098 & 1.81303 & 1.39792 \\ 
\end{tabular}

\end{table}

Tab.\ref{tab:07} shows that the temperature-jump coefficient
$\zeta_2$ is insensitive to the IPL potentials. To improve
the accuracy of \eqref{eq:z2}, a possible way is to write
the expression of $\zeta_2$ when $M=5$. However, the expression
is too complicated to explicitly write down. 
An alternative method is to assume
\[
	\zeta_2 = \frac{d_1n_1+d_2n_2+d_3n_1n_2}{c_1n_1+c_2n_2+1},
\]
where $c_i$ and $d_i$ are fitting coefficients to be determined
by intermolecular potentials. 

The data fitting is also promising to improve the accuracy of
\eqref{eq:z0m4} and \eqref{eq:z1-1}. From this point of view,
the analytical solutions to the moment model are instructive
and help us find the critical parameters that affect slip and
jump coefficients, i.e., $m_i$ and $n_i$ in \eqref{eq:defm} and
\eqref{eq:defn}.
It may deserve a detailed discussion somewhere else about 
these fitting formulae compared with the experimental and MD data. 
The issue should be future work and beyond the scope of this paper.

	\section{Conclusions}
	Our main intention was to develop the moment method of arbitrary
order for describing rarefied gas effects due to the general
gas-surface interaction. Utilizing the Hermite expansion, 
we derived the layer equations with general boundary conditions
in the frame of the moment method. These moment systems are
proved solvable after a simple boundary stabilization.
In particular, we discussed the Cercignani-Lampis scattering
kernel and gave a recursion formula to calculate its Hermite
expansion. This 
procedure avoids numerical integration
and helps us find explicit expressions for slip and jump
coefficients in terms of the accommodation coefficients.

Based on the moment model, we analyzed and evaluated
viscous-slip, thermal-slip, and temperature-jump coefficients for the 
inverse-power-law intermolecular potentials and 
Cercignani-Lampis boundary conditions. 
As shown in numerical tests, our moment model can capture
slip and jump coefficients accurately and efficiently
with mild moments. For low-order moment models,
explicit expressions of slip and jump coefficients
about the accommodation coefficients were derived.
These formulae are nonlinear, accurate, and concise in form, which
successfully explain some reported effects of the 
accommodation coefficients 
and intermolecular potentials.

	\appendix
\section{Orders of multi-indices}
\label{app:A}
\begin{defn}\label{def:order}
	We define the ordering $\preceq$ on $\bbN^D$ as follows.
	For $\+\alpha,\+\beta\in\bbN^D$, 
\begin{enumerate}
	\item If $\alpha_2$ is even and $\beta_2$ is odd, then 
		$\+\alpha \preceq \+\beta$.
	\item If $\alpha_2$ and $\beta_2$ have the same parity, but
		$|\+\alpha|<|\+\beta|$, then $\+\alpha \preceq
		\+\beta$.
	\item If $\alpha_2$ and $\beta_2$ have the same parity and
		$|\+\alpha|=|\+\beta|$, but there exists a smallest 
		$1\leq i\leq D$ such that $\alpha_i\neq\beta_i$, then
		$\+\alpha \preceq \+\beta$ if and only if 
		$\alpha_i \geq \beta_i$. 
	\end{enumerate}
\end{defn}
As usual, $\+\alpha\prec\+\beta$ means $\+\alpha\preceq\+\beta$
and $\+\alpha\neq\+\beta.$ In the above definition, a special
feature is that the indices with an even second component
are always ordered before the odd ones, e.g., the index
$(a_1,0,a_3)$ is ordered before $(b_1,1,b_3)$ for 
any $a_1,a_3,b_1$ and $b_3$. Except for that point, 
the multi-indices are first sorted by the multi-index norm and
then by the anti-lexicographic order.

If $\mathbb{I}$ is a subset of $\bbN^D$, then the ordering of
$\mathbb{I}$ is naturally defined as the restriction of $\preceq$
to $\mathbb{I}$. If $\mathbb{I}\subset\bbN^D$ is finite
with $\#\mathbb{I}$ elements, then
$\mathbb{I}$ is isomorphic to $\{1,2,\cdots,\#\mathbb{I}\}$.
In this paper, we define the default isomorphism 
$\mathcal{N}:\mathbb{I}\rightarrow \{1,2,\cdots,\#\mathbb{I}\}$ 
by the ordering $\preceq$, i.e., 
$\tn{\+\alpha}\leq\tn{\+\beta}$ if and only if $\+\alpha\preceq
\+\beta.$

Here, if a vector $\+w$ is called induced from $\mathbb{I}$,
we mean that the length of $\+w$ is $\#\mathbb{I}$ and we use
$\+w[\+\alpha]$ to represent its $\tn{\+\alpha}$-th element
where $\+\alpha\in\mathbb{I}.$ Analogously, if a matrix $\+A$
is called induced from $\mathbb{I}_1\times\mathbb{I}_2$, then
the size of $\+A$ is $(\#\mathbb{I}_1)\times(\#\mathbb{I}_2)$.
For $\+\alpha\in\mathbb{I}_1$ and $\+\beta\in\mathbb{I}_2$,
we use $\+A[\+\alpha,\+\beta]$ to represent its entry in the
$\tni{\+\alpha}$-th row and $\tnj{\+\beta}$-th column, where
$\mathcal{N}_i,i=1,2,$ are default isomorphic functions 
for $\mathbb{I}_i.$ 
 
\section{Calculation of $S_0(\alpha_2,\beta_2)$}
\label{app:B}
From \eqref{eq:def_S0} and the recursion relation 
\eqref{eq:defH}, we have
\begin{eqnarray}\label{eq:aS}
	S_0(\alpha_2,\beta_2) 
	&=& \int_0^{+\infty}\!\!
	\left(\sqrt{\alpha_2}\phi_{\alpha_2-1}+\sqrt{\alpha_2+1}
	\phi_{\alpha_2+1}\right)\phi_{\beta_2}
	\omega_0(\xi_2)\, \mathrm{d}\xi_2 \\
	\notag
	&=& \sqrt{\alpha_2}I(\alpha_2-1,\beta_2)
	+\sqrt{\alpha_2+1}I(\alpha_2+1,\beta_2),
\end{eqnarray}
where we define
\begin{eqnarray*}
I(\alpha_2,\beta_2) = \int_0^{+\infty}\!\!
	\phi_{\alpha_2} \phi_{\beta_2} \omega_0(\xi_2)
   	\mathrm{d}\xi_2.
\end{eqnarray*}

According to integration by parts, noting that
\begin{equation}
	\label{eq:inp}
	\od{}{\xi_2}\left(\omega_0(\xi_2)\phi_{\beta_2}\right)=
	-\sqrt{\beta_2+1}\omega_0(\xi_2)\phi_{\beta_2+1},\quad
	\od{}{\xi_2}\phi_{\alpha_2+1} =
	\sqrt{\alpha_2+1}\phi_{\alpha_2},
\end{equation}
we have
	\begin{eqnarray}
		\label{eq:BI1}
	I(\alpha_2+1,\beta_2+1) = 
		\left(\phi_{\alpha_2+1}(0)
			   	\phi_{\beta_2}(0)\omega_0(0)+
		\sqrt{\alpha_2+1}I(\alpha_2,\beta_2)\right)/\sqrt{\beta_2+1}. 
	\end{eqnarray}
The symmetry gives $I(\alpha_2,\beta_2)=I(\beta_2,\alpha_2)$.
So we also have
	\begin{eqnarray}
		\label{eq:BI2}
	I(\alpha_2+1,\beta_2+1) = 
		\left(\phi_{\beta_2+1}(0)
			   	\phi_{\alpha_2}(0)\omega_0(0)+
		\sqrt{\beta_2+1}I(\alpha_2,\beta_2)\right)/\sqrt{\alpha_2+1}. 
	\end{eqnarray}

When $\alpha_2\neq\beta_2$, we can equal \eqref{eq:BI1} and
\eqref{eq:BI2} to obtain the explicit expression 
of $I(\alpha_2,\beta_2).$ For simplicity, we denote by
$z_{n}=\phi_{n}(0)$. Then the recursion relation \eqref{eq:defH}
gives 
\[
	z_{n+1}=-\sqrt{n}z_{n-1}/\sqrt{n+1}
\]
with $z_0=1$ and $z_1=0$. We also have
$\omega_0(0)=(\sqrt{2\pi})^{-1}$.
When $\alpha_2=\beta_2$, noting that $z_n=0$ when $n$ is odd,
we have $I(\alpha_2+1,\alpha_2+1)=I(\alpha_2,\alpha_2)$
with $I(0,0)=1/2.$
From \eqref{eq:aS}, we can finally write the expressions of
$S_0(\alpha_2,\beta_2).$

\section{Calculation of $T(\alpha_1,\beta_1)$}
\label{app:C}
Assume the Fubini theorem holds such that we can exchange the
integration order. We will calculate $T(\alpha_1,\beta_1)$ by
induction. Let $\alpha_1=0$, then the orthogonality of Hermite
polynomials gives
\begin{eqnarray*}
	T(0,\beta_1) = \int_{\bbR}\!\!\phi_{\beta_1}(\xi_1')
	\omega_0(\xi_1')\,\mathrm{d}\xi_1'
	= \delta_{\beta_1,0}.
\end{eqnarray*}
Integrate by parts with \eqref{eq:inp}
and use the recursion relation \eqref{eq:defH},
then we have
\begin{eqnarray*}
	Y_{\alpha_1} &\triangleq&
\int_{\bbR}\!\!\phi_{\alpha_1}(\xi_1)
   	\exp\left(-\frac{|\xi_1-(1-\alpha_t)\xi_1'|^2}
			{2\alpha_t(2-\alpha_t)}\right)\,\mathrm{d}\xi_1 \\
	&=& \frac{1}{\sqrt{\alpha_1+1}}
\int_{\bbR}\!\!\frac{\xi_1-(1-\alpha_t)\xi_1'}
{\alpha_t(2-\alpha_t)}\phi_{\alpha_1+1}(\xi_1)
   	\exp\left(-\frac{|\xi_1-(1-\alpha_t)\xi_1'|^2}
			{2\alpha_t(2-\alpha_t)}\right)\,\mathrm{d}\xi_1 \\
	&=& \frac{1}{\sqrt{\alpha_1+1}}
\int_{\bbR}\!\!\frac{-(1-\alpha_t)\xi_1'\phi_{\alpha_1+1}+
\sqrt{\alpha_1+1}\phi_{\alpha_1}+\sqrt{\alpha_1+2}\phi_{\alpha_1+2}}
{\alpha_t(2-\alpha_t)}
   	\exp\left(-\frac{|\xi_1-(1-\alpha_t)\xi_1'|^2}
			{2\alpha_t(2-\alpha_t)}\right)\,\mathrm{d}\xi_1
\end{eqnarray*}
Collecting the terms, we have
\begin{eqnarray}
	\label{eq:recY}
\sqrt{\alpha_1}Y_{\alpha_1} = -(1-\alpha_t)^2\sqrt{\alpha_1-1}Y_{\alpha_1-2}
+(1-\alpha_t)\xi_1'Y_{\alpha_1-1},\quad \alpha_1\geq 2,
\end{eqnarray}
with $Y_0=1$ and $Y_1=(1-\alpha_t)\xi_1'.$
Using the recursion relation \eqref{eq:defH} to get rid of
$\xi_1'$ in \eqref{eq:recY} and recalling the definition 
of $T(\alpha_1,\beta_1)$, we have the recursion formula
\begin{eqnarray}
	\label{eq:recT}
	T(\alpha_1,\beta_1) &=& -(1-\alpha_t)^2\frac{\alpha_1-1}{\alpha_1}
	T(\alpha_1-2,\beta_1) \notag \\ &&
	+(1-\alpha_t)\frac{\sqrt{\beta_1}T(\alpha_1-1,
		\beta_1-1)+\sqrt{\beta_1+1}T(\alpha_1-1,\beta_1+1)}
{\sqrt{\alpha_1}}.
\end{eqnarray}
Here we let $T(-1,\beta_1)=0$ and get
\[
	T(1,\beta_1) = \int_{\bbR}\!\!\phi_{\beta_1}(\xi_1')
	\omega_0(\xi_1')Y_1\,\mathrm{d}\xi_1'
	= (1-\alpha_t)\delta_{\beta_1,1}.
\]
By induction, we can see that 
$T(\alpha_1,\beta_1)=0$ when $\beta_1>\alpha_1$. 

However, by definition, we have the symmetry 
\[
	T(\alpha_1,\beta_1) = T(\beta_1,\alpha_1).
\]
So $T(\alpha_1,\beta_1)=0$ when $\alpha_1>\beta_1,$
too. Thus, the recursion relation \eqref{eq:recT} becomes
\[
T(\alpha_1,\alpha_1) = 	(1-\alpha_t)T(\alpha_1-1,\alpha_1-1).
\]
We finally have $T(\alpha_1,\beta_1)=\delta_{\alpha_1,\beta_1}
(1-\alpha_t)^{\alpha_1}.$

\section{Proof of Lemma \ref{lem:03}}
\label{app:S}
Firstly, we introduce some preliminaries.
The series expansion of the zeroth order
modified Bessel function (8.447 in Ref.\onlinecite{GR2014}) is
\[
	I_0\left(\frac{\sqrt{1-\alpha_n}}{\alpha_n}\xi_2\xi_2'\right)
	= \sum_{k=0}^{\infty} \frac{(1-\alpha_n)^k}{(2\alpha_n)^{2k}}
	\frac{1}{(k!)^2}\xi_2^{2k}\xi_2'^{2k}.
\]
Denote by
\[
	J(\alpha_2,k) = \frac{1}{\alpha_n^{k+1}(2k)!!}
	\int_0^{+\infty}\!\!\xi_2^{2k+1}
	\exp\left(-\frac{\xi_2^2}{2\alpha_n}\right)\phi_{\alpha_2}(\xi_2)
	\,\md\xi_2.
\]
Then we can write
\begin{equation}
	\label{eq:NconJ}
	N(\alpha_2;\xi_2') = \frac{1}{\sqrt{2\pi}}
	\sum_{k=0}^{\infty}\frac{1}{2^kk!}\left(
	\frac{1-\alpha_n}{\alpha_n}\right)^kJ(\alpha_2,k)\xi_2'^{2k}.
\end{equation}
On the other hand, the power series expansion of the
exponential function gives
\begin{equation}
	\label{eq:seH}
	\mH(\xi_2')\triangleq
	\frac{1}{\sqrt{2\pi}}\exp\left(-\frac{1}{2}
	\left(1-\frac{1}{\alpha_n}\right)\xi_2'^2\right)
	= \frac{1}{\sqrt{2\pi}}\sum_{k=0}^{\infty}
	\frac{1}{k!}\left(\frac{1-\alpha_n}{2\alpha_n}\right)^k\xi_2'^{2k}.
\end{equation}

Secondly, we calculate $J(\alpha_2,k)$ for even $\alpha_2$.
According to the formula 7.376 in Ref.\onlinecite{GR2014}, we have
\[
	J(2n,k) =
	\frac{(-1)^n2^n\Gamma(n+\frac{1}{2})}{\sqrt{\pi}\sqrt{(2n)!}}
	F(-n,k+1;\frac{1}{2},\alpha_n),
\]
where $n,k\in\bbN$. Here we use the Gamma function, i.e.,
\[
	\Gamma(n+\frac{1}{2}) = (n-\frac{1}{2})(n-\frac{3}{2})\cdots
	\frac{1}{2}\sqrt{\pi}
\]
and the hypergeometric function, i.e., 
\[
	F(-n,k+1;\frac{1}{2},\alpha_n) = 1 
	+ \frac{(-n)}{\frac{1}{2}}C_{k+1}^1\alpha_n +
	\frac{(-n)(-n+1)}{\frac{1}{2}\frac{3}{2}}C_{k+2}^2
	\alpha_n^2 + \cdots + \frac{(-1)^n n!}
	{\frac{1}{2}\frac{3}{2}\cdots(\frac{1}{2}+n-1)}C_{k+n}^n
	\alpha_n^n,
\]
where the combination number $C_{k+n}^n=(k+1)\cdots(k+n)/n!.$
So we have $J(0,k)=1$. Let $J(-2,k)=0$.
According to Gauss's recursion 
functions (9.137 in Ref.\onlinecite{GR2014}), we have the following
recursion relation about $J(2n,k)$:
\begin{eqnarray} \notag 
	\sqrt{(2n+1)(2n+2)}J(2n+2,k) &=& 
	(\alpha_n-1)\sqrt{(2n)(2n-1)}J(2n-2,k) \\ && +
	(-4n-1+(2n+2)\alpha_n+2k\alpha_n)J(2n,k).
	\label{eq:recF}
\end{eqnarray}

Then we calculate $N(\alpha_2,\xi_2')$ by induction.
Since $J(0,k)=1$, by comparing \eqref{eq:NconJ} and \eqref{eq:seH},
we have
\[
	N(0;\xi_2') = \mH(\xi_2').
\]
Take the derivatives about $\xi_2'$ in \eqref{eq:seH}, and we have
\begin{eqnarray*}
	\xi_2'\od{}{\xi_2'}\mH(\xi_2') &=&
	-(1-\frac{1}{\alpha_n})\xi_2'^2\mH(\xi_2')\\
	&=& \frac{1}{\sqrt{2\pi}}\sum_{k=0}^{\infty}
	(2k)\frac{1}{k!}\left(\frac{1-\alpha_n}{2\alpha_n}
	\right)^k\xi_2'^{2k}.
\end{eqnarray*}
So combined with \eqref{eq:recF}, we have 
\[
	\sqrt{2}N(2;\xi_2') = (-1+2\alpha_n)N(0;\xi_2')
	- (\alpha_n-1)\xi_2'^2N(0;\xi_2').
\]
This inspires us to guess that $N(\alpha_2;\xi_2')$
can be represented as the product of $\mH(\xi_2')$ and
a polynomial of degree $\alpha_2$.
By induction principle, assuming that for any $m\leq n$,
there exist coefficients $r_{2m,2s}$ such that
\begin{equation}
	\label{eq:forN}
	N(2m;\xi_2') = \mH(\xi_2')\sum_{s=0}^{m} r_{2m,2s}
	\phi_{2s}(\xi_2'),
\end{equation}
where $\phi_{2s}(\xi_2')$ are Hermite polynomials as in 
Lemma \ref{lem:03}. Then taking the derivatives 
in \eqref{eq:NconJ} for $\alpha_2=2n$, we have
\begin{eqnarray*}
	\frac{1}{\sqrt{2\pi}}
	\sum_{k=0}^{\infty}\frac{1}{2^kk!}\left(
	\frac{1-\alpha_n}{\alpha_n}\right)^k(2k)J(2n,k)\xi_2'^{2k}
	&=& \xi_2'\od{}{\xi_2'}N(2n;\xi_2').
\end{eqnarray*}
Combined with the recursion relation \eqref{eq:recF} and
\eqref{eq:NconJ}, we have
\begin{eqnarray*}
	\sqrt{(2n+1)(2n+2)}N(2n+2,\xi_2')
	&=& \sqrt{(2n)(2n-1)}(\alpha_n-1)N(2n-2,\xi_2') \\
	&& +((2n+2)\alpha_n-4n-1)N(2n,\xi_2')
	+ \alpha_n\xi_2'\od{}{\xi_2'}N(2n;\xi_2').
\end{eqnarray*}
Hence, $N(2n+2,\xi_2')$ also has the form as in \eqref{eq:forN}.
Using \eqref{eq:inp}, the final recursion
relation of $r_{\alpha_2,\beta_2}$ is
\begin{eqnarray}
	&&\sqrt{(2n+1)(2n+2)}\sum_{s=0}^{n+1}r_{2n+2,2s}\phi_{2s}(\xi_2')
	\notag \\ \notag
	&=& \sqrt{(2n)(2n-1)}(\alpha_n-1)\sum_{s=0}^{n-1}
	r_{2n-2,2s}\phi_{2s}(\xi_2') +((2n+2)\alpha_n-4n-1)\sum_{s=0}^n
	r_{2n,2s}\phi_{2s}(\xi_2') \\ \label{eq:recR}
	&& + \sum_{s=1}^{n}r_{2n,2s}\sqrt{(2s)(2s-1)}\phi_{2s-2}(\xi_2')
	+ \sum_{s=0}^{n}r_{2n,2s}(2s)\phi_{2s}(\xi_2') \\
	&& +(1-\alpha_n)\sum_{s=0}^n r_{2n,2s}\left(
	(2s+1)\phi_{2s}(\xi_2')+\sqrt{(2s+1)(2s+2)}
	\phi_{2s+2}(\xi_2')\right). \notag
\end{eqnarray}

In conclusion, we have the following results:
(1). $N(2n,\xi_2')$ has the assumed form \eqref{eq:forN}.
(2). The recursion relation of $r_{2n,2s}$ is given
by comparing the coefficients before Hermite
polynomials in \eqref{eq:recR}, with the initial
values $r_{-2,2s}=0$ and $r_{0,2s}=\delta_{s,0}.$
(3). In particular, from \eqref{eq:recR},
the coefficients $r_{2n,2n}=(1-\alpha_n)^n.$

	\bibliography{../ATC}

\end{CJK*}
\end{document}